\documentclass[11pt]{article}
\usepackage[a4paper,margin=1in]{geometry}
\usepackage{amsmath,amssymb,amsthm,mathtools}
\usepackage{bm}
\usepackage{graphicx}
\usepackage{booktabs}
\usepackage{enumitem}
\usepackage[title]{appendix} 
\usepackage{algorithmicx}
\usepackage{algorithm}
\usepackage{algpseudocode}
\usepackage{float}

\usepackage{tikz}
\usetikzlibrary{arrows.meta, positioning}

\usepackage{authblk}

\usepackage{biblatex}
\usepackage[hidelinks]{hyperref}

\title{Stochastic compliance–evasion dynamics in tax models:\\ a piecewise deterministic Markov process approach}

\author[1]{Jonas Mayr\thanks{jonasmayr02@gmail.com}}
\author[1]{Amira Meddah\thanks{Corresponding author: amira.meddah@jku.at}}
\author[2]{Irene Tubikanec\thanks{irene.tubikanec@jku.at}}

\affil[1]{Institute of Stochastics, Johannes Kepler University Linz, Altenberger Straße 69, 4040 Linz, Austria}
\affil[2]{Institute of Applied Statistics, Johannes Kepler University Linz, Altenberger Straße 69, 4040 Linz, Austria}

\date{} 

\theoremstyle{plain}

\newtheorem{proposition}{Proposition}

\theoremstyle{definition}
\newtheorem{definition}{Definition}
\newtheorem{remark}{Remark}
\newtheorem{property}{Property}

\newcommand{\dd}{\,\mathrm{d}}

\begin{document}
\maketitle

\begin{abstract}
This paper introduces a novel stochastic framework for modelling tax evasion dynamics by extending the deterministic model of Bertotti and Modanese (2018) through the use of Piecewise Deterministic Markov Processes (PDMPs). A key limitation of the original model is the static treatment of taxpayer compliance and evasion behaviour. We address this limitation by incorporating two stochastic mechanisms: \textit{(i)} \emph{audits}, where random enforcement events shift non-compliant individuals toward compliance, and \textit{(ii)} \emph{imitation}, where social influence drives compliant individuals toward evasion. We develop each mechanism as a separate PDMP, proving that both preserve the fundamental conservation laws of population and global income. Numerical simulations show that these mechanisms produce opposing long-term outcomes: pure audits lead to full compliance, while pure imitation leads to full evasion. The central contribution is a combined PDMP model in which both dynamics interact. This model no longer converges to an extreme equilibrium state. Instead, it can exhibit persistent fluctuations around the deterministic trajectory and suggests convergence to a stationary distribution, providing a more realistic representation of compliance-evasion dynamics observed in real economies. The proposed framework offers a versatile approach for integrating behavioural stochasticity into socio-economic models.
\end{abstract}

\vspace{0.2cm}

\noindent\textbf{Keywords}\\ Piecewise deterministic Markov processes, Complex systems, Socio-economic models, Tax models, Tax evasion

\section{Introduction}

Tax evasion represents a persistent challenge for many economies and has significant macroeconomic and social consequences. By reducing tax revenues, evasion limits the capacity of public institutions to provide essential services such as education, health care, and infrastructure \mbox{\cite{chetty2009taxable, feldstein1995effect}}. Beyond its direct fiscal costs, tax evasion contributes significantly to increasing income inequality and impairs the redistributive capacity of the state among other consequences  \mbox{\cite{andreoni1998tax,biondo2022taxation, kirchler2007economic, slemrod2007cheating,slemrod2002tax, torgler2002speaking}}. The mechanisms underpinning evasion behaviour are multifaceted and inherently interactive. They result from a complex interplay of individual decisions, institutional frameworks, and policy measures. Several studies have consistently shown that the probability of tax evasion increases when individuals perceive non-compliance to be widespread, particularly when the perceived risk of detection is low \cite{dubin2007criminal, mcgee2011ethics, torgler2007tax}. Equally important are subjective factors, such as perceptions of fairness, levels of trust in government, and deeply held social and personal norms \cite{kirchler2007economic, shu2012retracted}.

A wide range of research approaches these phenomena through mathematical modelling, which can be broadly divided into two main directions. The first direction builds on the foundational work of Allingham and Sandmo \cite{allingham1972income}, who define tax evasion as a decision made under risk. In this model, a fully informed, rational agent optimises their expected utility based on audit probabilities and penalty rates. This paradigm has been later extensively refined and incorporates factors such as tax rates, income distribution, and alternative audit strategies~\mbox{\cite{andreoni1998tax,reinganum1985income,yitzhaki1974note}}.

A more recent, complementary approach uses agent-based and kinetic models to represent the complex, systemic nature of tax evasion. These models shift the focus from a single optimising agent to a population of interacting individuals who adjust their behaviour based on socioeconomic incentives and peer influence. This perspective is particularly well-suited to investigating how decentralised decisions lead to aggregated outcomes, such as the emergence of social norms and the dynamics of tax compliance in a heterogeneous population. Seminal contributions in agent-based economics paved the way for tax-specific applications \cite{aoki2004modeling, arthur2018economy}. Moreover, Bertotti and Modanese introduced in \cite{ bertotti2018mathematical,bertotti2014micro} a deterministic model that describes a closed society with income classes and tax evasion sectors and uses a nonlinear ordinary differential equation (ODE) to analyse long-term income distribution and tax compliance. Similar agent-based models \cite{bloomquist2012agent, hokamp2018agent, pickhardt2014behavioral} have been used to simulate the effects of policy measures and behavioural attributes in complex, evolving economies.

Building on the deterministic model introduced in \cite{bertotti2018mathematical}, this paper addresses a significant structural limitation: the static nature of individual evasion behaviour. In the original formulation, agents are assumed to remain within their compliance sectors, an assumption that, while analytically tractable, ignores the complex behavioural ecology of tax evasion. In reality, compliance is a dynamic decision, continuously shaped by institutional incentives like audit risks and penalties, as well as social factors such as peer effects and evolving norms.

To incorporate these dynamics, we extend the original model by formalising stochastic transitions between evasion sectors using Piecewise Deterministic Markov Processes (PDMPs). This class of stochastic processes was introduced by Davis in 1984 \cite{davis1984piecewise} and belongs to the general class of stochastic hybrid systems, which combine continuous deterministic or stochastic evolution with discrete stochastic jump dynamics \cite{blom1988piecewise, bujorianu2005toward, desmettre2025approximate, meddah2024stochastic}. The PDMP framework has been widely used to model complex systems across various fields, including biology, neuroscience, and finance \cite{buckwar2023stochastic, buckwar2011exact, buckwar2025american, desmettre2025hybrid}, to name a few.

Here, we illustrate that PDMPs also offer a powerful framework for modelling tax systems whose states evolve deterministically (via an ODE) until a random event, such as an audit or a change in social perception, triggers an instantaneous shift in tax compliance behaviour. This extension not only enhances the realism of the model but also enables the investigation of how policy parameters and social interactions co-evolve to create a macroscopic distribution of compliance and income.

To the best of our knowledge, this constitutes the first application of PDMPs to model behavioural dynamics and audit effects in tax evasion.  In this work, we introduce three PDMP based models that incorporate these effects:
\begin{enumerate}
    \item an \emph{audit PDMP}, in which random audit events induce transitions from
          non-compliance to compliance;
    \item an \emph{imitation PDMP}, in which individuals adopt the behaviour of economic
          peers, leading to transitions toward less compliant sectors;
    \item a \emph{combined PDMP}, in which audits and imitation act simultaneously,
          generating a dynamic balance between opposing behavioural forces.
\end{enumerate}
For each PDMP model, we specify its three characteristic components, namely the \textit{deterministic flow}, the \textit{jump rate}, and the \textit{transition kernel}, and
verify that fundamental structural properties of the original deterministic system \cite{bertotti2018mathematical}, such as conservation of total population and total income, are preserved. 
We complement the analytical results with a detailed simulation study illustrating the effects of behavioural mechanisms on population dynamics. Numerical experiments show that the audit and imitation mechanisms have opposite long-term tendencies, while their combination yields stationary patterns consistent with empirical observations of persistent partial compliance.

This paper is organised as follows. In Section~\ref{Back_sec} we recall the deterministic model introduced in \cite{bertotti2018mathematical}. Section~\ref{PDMP_intro} then recalls the theoretical framework of PDMPs, laying the mathematical groundwork for our model extensions. Building on this, Section~\ref{PDMP_form} develops our first key contribution: the formulation of PDMP models for audit and imitation dynamics, including an illustrative case study for each. Section~\ref{Combined_mod} presents our comprehensive combined model, detailing its structural properties with insights from numerical simulations. Conclusions and perspectives for
future work are summarised in Section~\ref{conclusion}. Finally, in Appendix \ref{PDMP_simulation}, we detail how to simulate the proposed PDMP models.

\section{Background and deterministic model}
\label{Back_sec}

In this section, we recall the original deterministic model of Bertotti and Modanese \cite{bertotti2018mathematical}. 

\paragraph{Setting.}

Consider a population of individuals divided into a finite number of income classes~$n$, with average incomes $0<r_1<\dots<r_n$ for each class, and a finite number of evasion sectors~$m$, representing different types of tax compliance or evasion behaviours. The population state is then described by
\begin{equation*}
    x=(x_j^\alpha)_{j=1,\dots,n;\,\alpha=1,\dots,m},
\end{equation*}
where 
$x_j^\alpha$ denotes the fraction of individuals in income class $j$ and evasion  sector $\alpha$, i.e., in group~$(j,\alpha)$.

In this framework, the society is assumed to be closed, in the sense that no individuals enter or leave the population and no external economic interactions occur. Each evasion behaviour of a taxpayer is assumed to remain constant over time, while individuals may change income class as a result of economic interactions. Taxation and redistribution are incorporated in the model through pairwise monetary exchanges, i.e., when two individuals interact, the receiving party is taxed according to the rate of their income class, and the collected amount is redistributed across the population. In the presence of evasion, only a fraction of the due taxes is actually paid, while the remaining part is concealed, reducing the resources available for redistribution.

These class-dependent taxation rules and sector-specific evasion behaviours determine the evolution of the fractions $x_j^\alpha$ over time. More precisely, when an $(h,\beta)$-agent meets a $(l,\vartheta)$-agent, for \mbox{$h,l \in \{1,\ldots,n\}$} and $\beta,\vartheta\in \{1,\ldots,m\}$, an amount of money $S>0$ is transferred from the agent in class $h$ to the agent in class $l$ with probability $p_{h,l}$ (payment probability), and taxation with redistribution is applied. The payment probability $p_{h,l}$ specifies the likelihood that, in an encounter
between an $h$-class and a $l$-class individual, the $h$-individual transfers the
amount $S$ to the $l$-individual. It satisfies
\begin{equation*}
0 \leq p_{h,l} \leq 1, 
\qquad p_{h,l} + p_{l,h} \leq 1 ,    
\end{equation*}
so that in each encounter either one individual pays, or no transfer occurs.
A common choice, as in \cite{bertotti2018mathematical}, is
\begin{equation}
 p_{h,l} = \frac{\min\{r_h,r_l\}}{4r_n}, 
 \label{paym_prob}
\end{equation}
together with the special cases summarised in Table \ref{tab:paymentProbs}.
\begin{table}
\centering
\begin{tabular}{ll}
\toprule
\textbf{Payment probability} & \textbf{Interpretation} \\
\midrule
$p_{j,j} = \dfrac{r_j}{2r_n}$ & self-class interactions \\[10pt]
$p_{h,1} = \dfrac{r_1}{2r_n}$ & payments to the poorest class \\[10pt]
$p_{n,l} = \dfrac{r_l}{2r_n}$ & payments from the richest class \\[10pt]
$p_{1,l} = 0$ & poorest individuals never pay \\[10pt]
$p_{h,n} = 0$ & richest individuals never receive \\
\bottomrule
\end{tabular}
\caption{Special cases of payment probabilities.}
\label{tab:paymentProbs}
\end{table}
\medskip
This formulation ensures heterogeneity in the frequency of payments across income classes and prevents unrealistic transfers at the income boundaries.

Each income class $j \in \{ 1,\ldots,n \}$ is associated with a statutory tax rate $\tau_j$, which represents the fraction of income that taxpayers in that class are officially required to pay. However, in the presence of tax evasion, only a part of this amount
is effectively paid. To model this, we introduce for each evasion sector $\alpha \in \{1,\dots,m\}$ a compliance parameter
\begin{equation*}
 \theta^{\mathrm{ev}}(\alpha) \in [0,1],   
\end{equation*}
denoting the fraction of due taxes actually paid by individuals in sector $\alpha$.
Accordingly, the effective tax rate for a $(j,\alpha)$-individual is given by
\begin{equation}
    \theta_{j,\alpha} = \theta^{\mathrm{ev}}(\alpha)\,\tau_j .
    \label{frac-pay}
\end{equation}
Thus, for honest taxpayers $\theta^{\mathrm{ev}}(\alpha)=1$, while values $\theta^{\mathrm{ev}}(\alpha)<1$ describe partial evasion.

The effective rates $\theta_{j,\alpha}$ enter the model through the taxation and redistribution mechanism: when an $(h,\beta)$-agent transfers an amount $S$ to a $(l,\vartheta)$-agent, the latter retains a fraction \mbox{$(1-\theta_{l,\vartheta})S$}, while the remainder $\theta_{l,\vartheta} S$ is collected as taxes and redistributed among the population (excluding the richest class). In this way, the heterogeneity across income classes and evasion sectors jointly determines the redistribution flows.

\paragraph{ODE model.}

Combining this mechanism with the payment probabilities introduced above yields an $n\times m$-dimensional ODE, where each fraction $x_j^\alpha$ evolves according to
\begin{equation}\label{population_eq}
\frac{\dd x_j^\alpha(t)}{\dd t}
= \underbrace{\sum_{h,l=1}^n \sum_{\beta,\vartheta=1}^m
\Big(
   C^{(j,\alpha)}_{(h,\beta);(l,\vartheta)}
   + T^{(j,\alpha)}_{(h,\beta);(l,\vartheta)}(x(t))
\Big) \, x_h^\beta(t) x_l^\vartheta(t)}_{\text{inflow: interaction gains}}
- \underbrace{x_j^\alpha(t) \sum_{l=1}^n \sum_{\vartheta=1}^m x_l^\vartheta(t).}_{\text{outflow: interaction losses}}
\end{equation}
The coefficient $C^{(j,\alpha)}_{(h,\beta);(l,\vartheta)}$ encodes the transitions associated with direct exchanges, while the function  $T^{(j,\alpha)}_{(h,\beta);(l,\vartheta)}(x(t))$ captures the redistribution effects. They are specified as follows. 

The coefficient $C^{(j,\alpha)}_{(h,\beta);(l,\vartheta)}$ represents the probability that an individual initially in group $(h,\beta)$, after interacting with an individual in $(l,\vartheta)$, transitions to group $(j,\alpha)$. For any fixed~$(h,\beta)$ and $(l,\vartheta)$, these coefficients satisfy the normalisation condition
\begin{equation*}
    \sum_{j=1}^{n}\sum_{\alpha=1}^{m} C^{(j,\alpha)}_{(h,\beta);(l,\vartheta)} =1.
\end{equation*}
In particular, for an individual in a given evasion sector, the non-zero coefficients are specified as~follows.
\begin{itemize}
\item[\emph{i}.] \textit{Move to a lower class.}
An individual moves from class $h = j+1$ down to class $j$ (within the same sector $\alpha$) when a payment they make reduces their income below the threshold $r_j$. The corresponding transition coefficient is given by
\begin{equation*}
C^{(j,\alpha)}_{(j+1,\alpha);(l,\beta)}
= p_{j+1,l}\frac{S(1-\theta_{l,\beta})}{r_{j+1}-r_j},
\end{equation*}
which is defined only for $j \leq n-1$ and $l \leq n-1$. Here $p_{j+1,\,l}$ is the probability that an individual in class $j+1$ pays an amount $S$ to a partner in class $l$, and $(1-\theta_{l,\beta})S$ is the net amount effectively lost by the payer, which depends on the recipient's evasion level. The denominator $r_{j+1}-r_j$ rescales this expected loss by the income gap between the classes.

\item[\emph{ii}.] \textit{Move to a higher class.}
An individual moves from class $h = j-1$ up to class $j$ when an income gain from a received payment is sufficient to cross the threshold $r_j$. The corresponding coefficient is
\begin{equation*}
C^{(j,\alpha)}_{(j-1,\alpha);(l,\beta)}
= p_{l,j-1}\frac{S(1-\theta_{j-1,\alpha})}{r_j-r_{j-1}},
\end{equation*}
defined only for $j \geq 2$ and $l \geq 2$. In this case, the $(j-1,\alpha)$-individual is the receiver, with probability $p_{l,\,j-1}$ that a partner in class $l$ pays them. The net income increment they retain, $(1-\theta_{j-1,\alpha})S$, depends on their own evasion behaviour.

\item[\emph{iii}.] \textit{Remain in the same class.}
If the income variation from an interaction does not cause a class change, an individual remains in their initial class $h = j$. The probability of this event happening is
\begin{equation*}
C^{(j,\alpha)}_{(j,\alpha);(l,\beta)}
= 1
- p_{l,j}\frac{S(1-\theta_{j,\alpha})}{r_{j+1}-r_j}
- p_{j,l}\frac{S(1-\theta_{l,\beta})}{r_j-r_{j-1}}.
\end{equation*}
In this expression:
\begin{itemize}
    \item[$\bullet$] The second summand (upward move term) is present only for $j\leq n-1$ and $l\geq 2$.
    \item[$\bullet$] The third summand (downward move term) is present only for $j\geq 2$ and $l\leq n-1$.
\end{itemize}
This formulation ensures that the net flow of individuals conserves the total population. The value $C^{(j,\alpha)}_{(j,\alpha);(l,\beta)}$ represents the propensity for an individual to remain in class $j$ after an interaction.
\end{itemize}

The function $T^{(j,\alpha)}_{(h,\beta);(l,\vartheta)} \colon \mathbb{R}^{n\times m} \to \mathbb{R}$ describes the variation in group $(j,\alpha)$ due to taxation and redistribution associated with a payment from an individual in $(h,\beta)$ to an individual in $(l,\vartheta)$. For any fixed $(h,\beta)$, $(l,\vartheta)$ and $x \in \mathbb{R}^{n\times m}$, these functions satisfy
\begin{equation*}
    \sum_{j=1}^{n} \sum_{\alpha=1}^{m} T^{(j,\alpha)}_{(h,\beta);(l,\vartheta)}(x) =0.
\end{equation*}
In particular, they are given by
\begin{equation}\label{eq:T-redistribution}
\begin{aligned}
T^{(j,\alpha)}_{(h,\beta);(l,\vartheta)}(x)
&=
\underbrace{%
\frac{p_{h,l}\,S\,\theta_{l,\vartheta}}{\displaystyle \sum_{i=1}^{n}\sum_{\nu=1}^{m} x_i^\nu}
\left(
\frac{x_{j-1}^{\alpha}}{\,r_j-r_{j-1}\,}
-\frac{x_{j}^{\alpha}}{\,r_{j+1}-r_j\,}
\right)
}_{\text{uniform redistribution}}
\\[1ex]
&\quad+
\underbrace{%
p_{h,l}\,S\,\theta_{l,\vartheta}
\left(
\frac{\delta_{h,j+1}\,\delta_{\alpha,\beta}}{\,r_h-r_j\,}
-\frac{\delta_{h,j}\,\delta_{\alpha,\beta}}{\,r_h-r_{j-1}\,}
\right)
\frac{\displaystyle \sum_{i=1}^{n-1}\sum_{\nu=1}^{m} x_i^{\nu}}
     {\displaystyle \sum_{i=1}^{n}\sum_{\nu=1}^{m} x_i^\nu}
}_{\text{boundary correction}}
\;,
\end{aligned}
\end{equation}
where $\delta_{\cdot,\cdot}$ denotes the Kronecker delta. The two terms appearing in \eqref{eq:T-redistribution} are interpreted as~follows.

\begin{itemize}
  \item[a.] The uniform redistribution term in \eqref{eq:T-redistribution} models the society-wide distribution of the tax revenue $p_{h,l}S\theta_{l,\vartheta}$. This revenue is normalised by the total population. The expression in parentheses calculates the net flow into class $j$: a fraction of the fund is gained from the class below ($j-1$) and loses a fraction to the class above ($j+1$). The explicit normalisation, while redundant under the condition $\sum_{i,\nu}x_i^\nu=1$, emphasises the per-capita nature of the redistribution.

  \item[b.] The boundary correction term in \eqref{eq:T-redistribution} ensures consistency for the payer when their payment causes a change in their own income class. The Kronecker deltas restrict this adjustment to the specific group the payer belongs to.
    \begin{itemize}
      \item[$\bullet$] The term $ +\dfrac{\delta_{h,j+1}\,\delta_{\alpha,\beta}}{\,r_h-r_j\,} $ applies if the payer moved from class $j+1$ down to $j$.
      \item[$\bullet$] The term $ -\dfrac{\delta_{h,j}\,\delta_{\alpha,\beta}}{\,r_h-r_{j-1}\,} $ applies if the payer moved from class $j$ down to $j-1$.
    \end{itemize}
  Further, the factor $\dfrac{\sum_{i=1}^{n-1} \sum_{\nu=1}^{m} x_i^\nu}{\sum_{i=1}^{n} \sum_{\nu=1}^{m} x_i^\nu}$ ensures this adjustment is only applied among the non-richest classes.
\end{itemize}

\paragraph{Existence of a unique solution and its properties.}

The $n\times m$-dimensional ODE defined by \eqref{population_eq} satisfies the following existence and uniqueness result, recalled from  \cite{bertotti2018mathematical}. 

\begin{property}[Existence of a unique solution]
 \label{prop0}
For any initial condition 
${x}_0 = (x^{\alpha}_{0j})_{j=1,\ldots,n;\alpha=1,\ldots,m}$ with 
\begin{equation*}
    x^{\alpha}_{0j} \geq 0, \quad \text{for all } j=1,\ldots,n, \ \alpha=1,\ldots,m,
\end{equation*}
and 
\begin{equation*}
    \sum_{j=1}^{n}\sum_{\alpha=1}^{m}x^{\alpha}_{0j}=1,
\end{equation*}
there exists a unique solution ${x}(t)=(x_j^\alpha(t))_{j=1,\ldots,n;\alpha=1,\ldots,m}$, defined for all $t \geq 0$ and satisfying~$x(0)=x_0$.
\end{property}

Moreover, this solution admits the following structural properties, also recalled from \cite{bertotti2018mathematical}. Properties \ref{prop1} and \ref{prop2}  are proved analytically (using arguments similar to those in \cite{BERTOTTI2010}), while Property \ref{prop3} is supported by numerical evidence.

\begin{property}[Population conservation]
\label{prop1}
Non-negativity is preserved, i.e.
\begin{equation*}
    x_j^\alpha(t) \geq 0, \quad \text{for all } j=1,\ldots,n, \ \alpha=1,\ldots,m, \text{ and } t \geq 0,
\end{equation*}
and the total population is conserved, i.e.
\begin{equation*}
\sum_{j=1}^{n}\sum_{\alpha=1}^{m} x_j^\alpha(t) = 1, \quad \text{for all } t \geq 0. 
\end{equation*}
\end{property}

\begin{property}[Global income conversation]
\label{prop2}
The global income remains constant over time, i.e.
\begin{equation*}
    \sum_{j=1}^{n} r_j \sum_{\alpha=1}^{m} x_j^\alpha(t)=\mu(x(t)) = \mu(x_0)=:\mu, \quad \text{for all } t \geq 0.
\end{equation*}
\end{property}

\begin{property}[Stationarity]
\label{prop3}
Solutions $x(t)$ evolving from initial states $x_0$ that, in addition to the conditions in Property \ref{prop0}, share the same initial ratio between evasion sectors across all income classes
and the same average initial income
\begin{equation*}
   \mu:=\mu(x_0)=\sum_{j=1}^{n} r_j \sum_{\alpha=1}^{m} x_{0j}^\alpha,
\end{equation*}
converge to the same stationary state $x^*$, i.e.
\begin{equation*}
    \lim\limits_{t \to +\infty} x(t)=x^*. 
\end{equation*}
\end{property}

\paragraph{An illustrative example and model simulation.}

Throughout this paper, we consider an illustrative tax system with $n=3$ income classes and $m=2$ evasion sectors.
The average incomes for the three income classes are set to
\begin{equation*}
r_1 = 10, \qquad r_2 = 20, \qquad r_3 = 30,    
\end{equation*}
\noindent creating uniform income gaps of $r_{j+1} - r_j = 10$. We adopt a progressive taxation scheme with statutory tax rates:
\begin{equation*}
   \tau_1 = 0.1, \qquad \tau_2 = 0.2, \qquad \tau_3 = 0.3. 
\end{equation*}
Moreover, we consider two evasion sectors representing distinct compliance behaviours:
\begin{itemize}
    \item Sector $\alpha=1$: \emph{fully compliant} taxpayers, with $\theta_{\mathrm{ev}}(1) = 1$,
    \item Sector $\alpha=2$: \emph{partial evaders}, with $\theta_{\mathrm{ev}}(2) = 0.5$.
\end{itemize}
Using~\eqref{frac-pay}, the resulting effective tax rates $\theta=(\theta_{j,\alpha})_{j=1,\ldots,n;\alpha=1,\ldots,m}$ are given by
\begin{equation*}
    \theta = 
\begin{pmatrix}
0.10 & 0.05\\
0.20 & 0.10\\
0.30 & 0.15
\end{pmatrix}.
\end{equation*}
Here, the transaction amount is set to $S = 1$, which satisfies the standard assumption ${S \ll r_{j+1} - r_j}$. The payment probabilities $p_{h,l}$, determined by~\eqref{paym_prob}, form the matrix
\begin{equation*}
 P =
\begin{pmatrix}
0 & 0 & 0\\[0.2ex]
{1}/{6} & {1}/{3} & 0\\[0.2ex]
{1}/{6} & {1}/{3} & 0
\end{pmatrix}.   
\end{equation*}
\noindent This structure reflects key model assumptions: individuals in the lowest class ($j=1$) never act as payers, while those in the highest class ($j=3$) never act as recipients.

Figure \ref{fig:ODE_example_setting} illustrates the corresponding ODE model \eqref{population_eq}. In particular, it shows the evolution of the population fractions $x_j^\alpha$, for $j=1,2,3$ and $\alpha=1,2$, obtained as solutions of the ODE for three different initial conditions $x_0$ satisfying the assumptions of Properties \ref{prop0} and \ref{prop3}. The initial conditions are
\begin{equation}
\label{initial_conditions}
    x_0^{(1)}=\begin{pmatrix}
        1/6 & 1/6 \\
        1/6 & 1/6 \\
        1/6 & 1/6
    \end{pmatrix}, \quad
    x_0^{(2)}=\begin{pmatrix}
        0 & 0 \\
        1/2 & 1/2 \\
        0 & 0
    \end{pmatrix}, \quad 
    x_0^{(3)}=\begin{pmatrix}
        {1}/{8} & {1}/{8} \\
        {1}/{4} & {1}/{4} \\
        {1}/{8} & {1}/{8}
    \end{pmatrix}, 
\end{equation}
corresponding to the blue, green, and orange curves, respectively. For each solution, Properties~\ref{prop1} and \ref{prop2} can be verified at each time $t$, and Property \ref{prop3} is visually apparent from the figure. Note that these simulations are based on a modified Euler method (see Appendix \ref{Im_Eul} for details).


\begin{figure}
    \centering
    \includegraphics[width=1.0\textwidth]{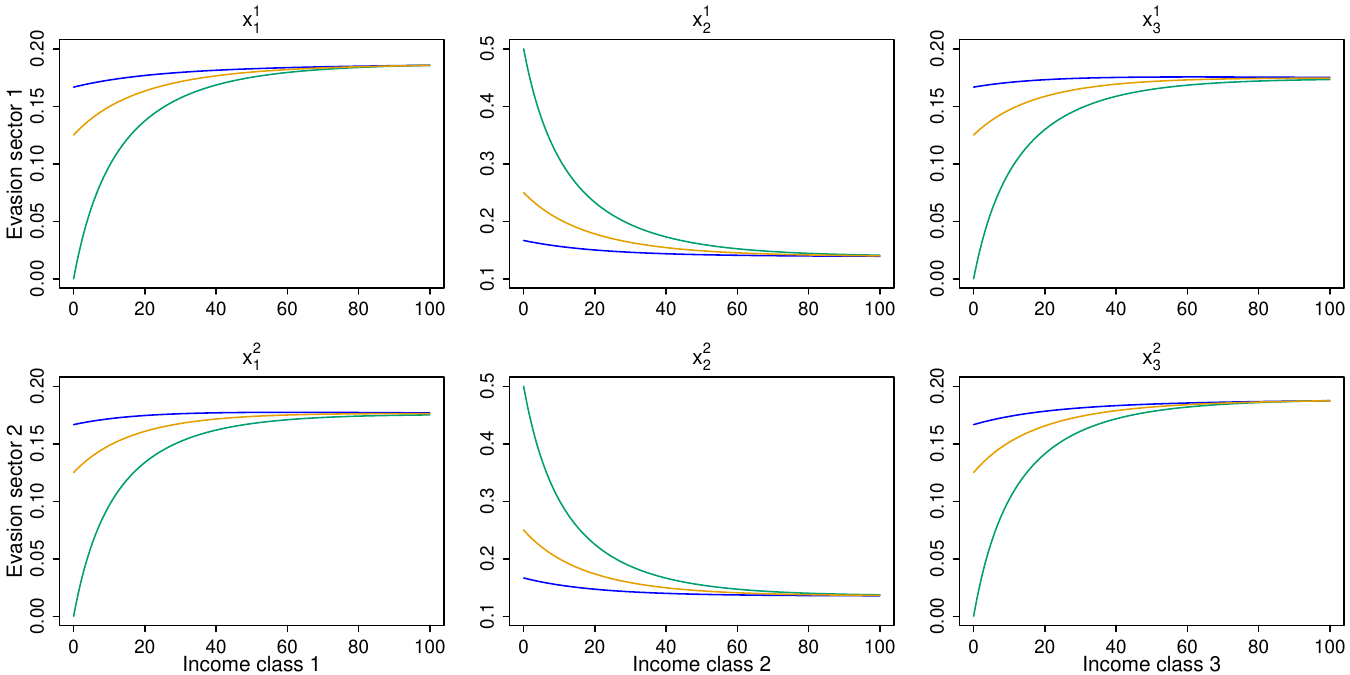}
    \caption{\textbf{Evolution of population fractions} for the example setting, using three different initial conditions $x_0$ (see \eqref{initial_conditions}) that satisfy the assumptions of Properties \ref{prop0} and \ref{prop3}.}
    \label{fig:ODE_example_setting}
\end{figure}

\section{Piecewise deterministic Markov processes}
\label{PDMP_intro}

In this section, we recall the general definition of PDMPs in the sense of Davis~\cite{davis1984piecewise}, which will be applied for tax modelling in the subsequent sections. A PDMP is a continuous-time Markov process whose randomness enters through state-dependent jump times, while between jumps the trajectories evolve deterministically according to an ODE. It is characterised by three components: the \emph{deterministic flow} $\phi$, the \emph{jump rate function} $\lambda$, and the \emph{transition kernel} $Q$.

Let $E = \mathcal{X} \times K$ be the state space, endowed with the product Borel $\sigma$-algebra, where $\mathcal{X}$ is a subset of $\mathbb{R}^d$ and $K$ is a finite set or a countable set. An element of $E$ is denoted by $z = (x,k)$, where $x \in \mathcal{X}$ represents the continuous component and $k \in K$ a discrete mode. We consider a stochastic process $Z(t)=( X(t),K(t) )$, for ${t\geq 0}$, taking values in $E = \mathcal{X} \times K$. The corresponding characteristic triplet $(\phi,\lambda,Q)$ defining the PDMP is introduced as follows.

\paragraph{Deterministic flow.}

For each mode $k \in K$, let $v_k : \mathcal{X} \to \mathbb{R}^d$ be a locally Lipschitz vector field. This field generates a unique global flow $\phi_k(t,x)$ determined by the ODE
\begin{equation*}
   \frac{d}{dt}\phi_k(t,x) = v_k(\phi_k(t,x)), 
\qquad \phi_k(0,x) = x. 
\end{equation*}
In other words, while the process remains in mode $k$, the continuous component $X(t)$ evolves deterministically according to the flow $\phi_k$, and the mode component $K(t)$ remains constant between jumps.

\paragraph{Jump mechanism.}

The transitions between states are governed by the jump mechanism, which is determined by two local characteristics:
\begin{enumerate}
    \item The \emph{jump rate function} $\lambda : E \to [0,\infty)$, assumed to be measurable, which determines the state-dependent intensity of jumps. For each $z=(x,k) \in E$, we assume that there exists $t>0$ such that 
    \begin{equation*}
        \int_0^{t} \lambda(\phi_k(s,x), k) \, ds < \infty.
    \end{equation*}
    \item The \emph{transition kernel} $Q(x,\cdot)$, which assigns to each current state $z=(x,k) \in E$ a probability measure on $(E,\mathcal{B}(E))$, specifying the distribution of the post-jump state. The kernel satisfies the following properties:
    \begin{itemize}
        \item For each fixed $A \in \mathcal{B}(E)$, the map $z \mapsto Q(z,A)$ is measurable.
        \item $Q(z,\{z\})=0$ for all $z$, ensuring that every jump changes the state.
    \end{itemize}
\end{enumerate}

\paragraph{PDMP construction.}

The dynamics of a PDMP are constructed by alternating deterministic motion with random jumps. More precisely, starting from $Z(0)=z_0=(x_0, k_0) \in E$, the trajectory is built recursively as follows.

Conditional on $Z(0)=(x_0,k_0)$, the first waiting time $\omega_1$ has a survival function given by
\begin{equation*}
 \mathbb{P}(\omega_1 > t \mid Z(0)) 
= \exp\!\left(-\int_0^t \lambda\big(\phi_{k_0}(s, x_0), k_0\big)\,ds\right), 
\qquad t \ge 0.   
\end{equation*}
On the interval $[0, T_1)$, where $T_1 = \omega_1$, the process evolves deterministically according to the flow generated by the vector field with mode $k_0$, i.e.,
\begin{equation*}
  Z(t) = (X(t),K(t)) = \big( \phi_{k_0}(t, x_0), k_0\big), \qquad t \in [0, T_1).
\end{equation*}

\noindent At the jump time $T_1$, the process transitions from the pre-jump state $Z({T_1^-})$ to a new state drawn from the transition kernel, i.e.,
\begin{equation*}
  Z({T_1}) \sim Q\big(Z({T_1^-}), \cdot\big).
\end{equation*}
The post–jump state always differs from the pre–jump state, meaning that, depending on the kernel $Q$, a jump may alter the continuous component, the discrete mode, or both.

Starting from the new initial state $z_1=Z({T_1})$, this procedure is repeated. At each step, conditional on the state immediately after the $(i-1)$-th jump, a waiting time $\omega_i$ is drawn according to the jump distribution with a jump rate $\lambda$. The process follows the deterministic flow until $T_i = T_{i-1} + \omega_i$, and a new post-jump state $z_i=Z({T_i})$ is sampled from $Q(Z({T_i^-}), \cdot)$. This yields an increasing sequence of jump times
\begin{equation*}
  0 = T_0 < T_1 < T_2 < \cdots,
\end{equation*}
together with a c\`adl\`ag\footnote{A c\`adl\`ag process is a stochastic process where paths are right-continuous with left limits everywhere, with probability one.} process $Z(t)=(X(t),K(t))$, for ${t \geq 0}$, that is piecewise deterministic. Specifically, $K(t)$ is piecewise constant and $X(t)$ evolves according to the respective ODE flows. 

In summary, the PDMP is given by
\begin{equation*}
    Z(t)=(X(t),K(t))=\big( \phi_{k_i}(t-T_i, x_i), k_i\big), \qquad t \in [T_i, T_{i+1}), \quad i\geq 0,
\end{equation*}
where $T_0:=0$, $(x_0, k_0)=z_0=Z(0)$, and
\begin{equation*}
 (x_{i},k_{i})=z_{i}=Z({T_{i}}) \sim Q\big(Z({T_{i}^-}), \cdot\big), \qquad i \geq 1.
\end{equation*}

To ensure this process is well-defined, we further assume that it is non-explosive, i.e., the number of jumps in any finite time interval is almost surely finite. That is, for all $t \ge 0$,
\begin{equation*}
  N_t = \sum_{i \ge 1} \mathbf{1}_{\{T_i \le t\}} < \infty \quad \text{a.s}.
\end{equation*}

\begin{definition}
A stochastic process $Z(t)=(X(t),K(t))$, ${t \geq 0}$, with state space $E = \mathcal{X} \times K$ is called a PDMP if it is constructed by the above procedure from the triplet of local characteristics~$(\phi, \lambda, Q)$.
\end{definition}

\begin{remark}
In the general formulation of Davis~\cite{davis1984piecewise}, the continuous domain $\mathcal{X}$ may depend on the mode $k$, and jumps may also be triggered upon hitting the boundary of $\mathcal{X}$. For our purposes, we restrict attention to the subclass where $\mathcal{X}$ is fixed across all modes and the deterministic flows remain in $\mathcal{X}$.
\end{remark}

\section{PDMP modelling of different sector transitions}
\label{PDMP_form}

In the ODE model of Bertotti and Modanese~\cite{bertotti2018mathematical} (cf. Section \ref{Back_sec}), one of the key
assumptions is that each individual’s evasion behaviour remains fixed over time. As a consequence, the population fraction in each evasion sector is constant, even though individuals may move between income classes. While this simplifies the analysis, it excludes important mechanisms that drive compliance in practice.

To address this limitation, we extend the model by allowing probabilistic transitions between sectors, formulated using PDMPs as introduced in Section \ref{PDMP_intro}. At random points in time, individuals may change their evasion behaviour due to events that affect their compliance. In particular, in this work we consider two mechanisms: audits, which may induce a taxpayer to become more compliant, and behavioural imitation,
through which individuals adopt the evasions strategies of others. The effects of these two mechanisms are first analysed separately in Sections \ref{audit_sec} and \ref{Sec_behave} below (with numerical experiments in Section \ref{ex:3x2}) and then combined into a joint PDMP model in Section \ref{Combined_mod}.

\subsection{Audits}
\label{audit_sec}

Audits represent inspections of individual financial records by an independent authority (e.g. the tax office). Their purpose is to verify the accuracy of reported income and the corresponding tax payments. In practice, audits occur on a recurring but unpredictable basis, which makes them naturally suited to be modelled as random events within the PDMP framework.

When tax evasion is detected, consequences may range from warnings to severe financial penalties or even criminal prosecution, typically accompanied by an invoice for the unpaid taxes. Such sanctions are also likely to increase the probability of further audits in the future. From a modelling perspective, these effects can be captured by assuming that audits induce a shift in compliance behaviour, that is, after being audited, individuals in non-compliant sectors ($\alpha>1$) transition with positive probability into the fully compliant sector ($\alpha=1$). This mechanism introduces stochastic sectoral movement that fundamentally extends the deterministic income class dynamics of the baseline ODE model.

\paragraph{Mathematical formulation of the audit PDMP.}

In the general framework of Section~\ref{PDMP_intro}, a PDMP is represented as a couple of stochastic processes $Z(t) = (X(t),K(t))$, $t\geq 0$, where $X(t)$ denotes the continuous component and $K(t)$ the discrete mode. 
In the present model the mode does not change, so $K(t)$ is constant and thus can be omitted for simplicity of notation. We therefore identify the PDMP with its continuous component and have 
\begin{equation*}
  Z(t) = X(t) = \bigl(X_j^\alpha(t)\bigr)_{j=1,\dots,n;\,\alpha=1,\dots,m}  = \phi(t-T_i, x_i), \qquad t \in [T_i,T_{i+1}), \quad i\geq 0,
\end{equation*}
where $\phi$ is the flow of the ODE defined by \eqref{population_eq}, $T_0:=0$, and the $T_i$, $i\geq 1$, correspond to random audit event times. This means that $X(t)$ follows the $n\times m$-dimensional ODE \eqref{population_eq} piecewise on each inter-event interval $[T_i,T_{i+1})$, with a different initial condition $x_i=X(T_i)$ for each ``piece''. These varying initial conditions reflect the population reallocation movements happening at each random event time $T_i$.  

The random audit events in this PDMP framework occur according to a state-dependent intensity. Specifically, the intensity function is chosen to scale with the share of non-compliant agents as follows
\begin{equation}\label{audit_rate}
  \lambda_{\text{audit}}(x) \;=\; \gamma \sum_{j=1}^n \sum_{\alpha=2}^m x_j^\alpha, 
\end{equation}
where $\gamma>0$ is a policy parameter controlling the overall level of enforcement. Thus, the higher the fraction of evaders in the population, the more frequently audits are expected to occur. 

The constant rate specification $\lambda_{\text{audit}}(x)\equiv \gamma_0>0$ is obtained as a special case of \eqref{audit_rate}, where the dependence on the non-compliant share is suppressed. This regime represents surprise audits, in which enforcement events occur at random times according to a homogeneous Poisson process, independently of the current level of evasion.

At each audit time $T_i$, a fixed proportion $\delta \in (0,1]$ of the non-compliant population is instantaneously reassigned to full compliance within each income class. In this setting, an individual's income class remains unchanged during these changes. Formally, the audit update is defined by the map $\Psi$, given component-wise by
\begin{equation}\label{audit_map}
  \Psi^\alpha_j(x) =
  \begin{cases}
    x^1_j + \delta \displaystyle\sum_{\beta=2}^m x^\beta_j, & \alpha = 1, \\[1.5ex]
    (1-\delta)\,x^\alpha_j, & \alpha > 1,
  \end{cases}
  \qquad j = 1, \dots, n.
\end{equation}
Thus, within each class $j$, the evading shares are reduced by a factor $(1-\delta)$ and the released mass $\delta \sum_{\beta=2}^{m} x_j^\beta$ is transferred to the compliant sector ($\alpha=1$). 
The corresponding transition kernel is deterministic,
\begin{equation*}
   Q_{\text{audit}}(x,\cdot) = \boldsymbol{\delta}_{\Psi(x)}(\cdot), 
\end{equation*}
and satisfies $Q_{\text{audit}}(x,\{x\})=0$ whenever 
at least one class contains a positive fraction of evaders.

\begin{remark}
    Stochastic extensions of the transition kernel are possible, for instance by choosing~$\delta$ randomly in $(0,1]$ at each audit event time.
\end{remark}

\paragraph{Well-posedness of the audit PDMP.}

Define the set
\begin{equation}
  \mathcal{X} := \Bigl\{\, x \in \mathbb{R}^{n \times m} \colon x_j^\alpha \geq 0, \
   \sum_{j=1}^n \sum_{\alpha=1}^m x_j^\alpha = 1 \,\Bigr\}. 
   \label{State_space}
\end{equation}
Then, the following holds for the transition map $\Psi$ as in \eqref{audit_map}.

\begin{proposition}\label{prop:audit}
For every $\delta \in (0,1]$ and $x \in \mathcal{X}$, the audit update $\Psi$ \eqref{audit_map} satisfies $\Psi(x)\in \mathcal{X}$. 
\end{proposition}

\begin{proof}
Let $\delta \in (0,1]$ and $x \in \mathcal{X}$. By construction, it holds that 
\begin{equation*}
    \Psi^\alpha_j(x) \geq 0, \quad \text{for all } j=1,\ldots,n \ \text{and} \ \alpha=1,\ldots,m.
\end{equation*}
Moreover, for each income class $j \in \{1,\ldots,n\}$, we have
\begin{align*}
 \sum_{\alpha=1}^m \Psi_j^\alpha(x)
&= \left(x_j^1 + \delta\sum_{\beta=2}^m x_j^\beta\right) + \sum_{\alpha=2}^m (1-\delta)x_j^\alpha \\
&= x_j^1 + \delta\sum_{\beta=2}^m x_j^\beta + \sum_{\alpha=2}^m x_j^\alpha - \delta\sum_{\alpha=2}^m x_j^\alpha \\
&= x_j^1 + \sum_{\alpha=2}^m x_j^\alpha \\
&= \sum_{\alpha=1}^m x_j^\alpha.
\end{align*}
Therefore, since $x\in \mathcal{X}$, and thus $\sum_{j=1}^n\sum_{\alpha=1}^m x_j^\alpha=1$, we also have that $\sum_{j=1}^n\sum_{\alpha=1}^m \Psi_j^\alpha(x)=1$. This implies the statement.
\end{proof}

Note that, Proposition \ref{prop:audit} readily implies that
\begin{equation*}
\sum_{j=1}^n r_j \sum_{\alpha=1}^m \Psi_j^\alpha(x) = \sum_{j=1}^n r_j \sum_{\alpha=1}^m x_j^\alpha.
\end{equation*}
Therefore, both population conservation (cf. Property \ref{prop1}) and global income conservation (cf.~Property~\ref{prop2}) are preserved by the audit update map $\Psi$ \eqref{audit_map}.

Proposition \ref{prop:audit} also implies that the audit PDMP is well-defined. Specifically, starting from $x_0 \in \mathcal{X}$, there exists a unique solution 
\begin{equation*}
    X(t)=\phi (t,x_0), \qquad t \in [0,T_1),
\end{equation*}
of the $n\times m$-dimensional ODE defined by \eqref{population_eq}, due to Property \ref{prop0}. At the random  audit event time $T_1$, the audit update map $\Psi$ \eqref{audit_map} determines the next initial condition 
\begin{equation*}
    x_1=X(T_1)=\Psi (X(T_1^-)),
\end{equation*}
which is an element of $\mathcal{X}$ due to Proposition \ref{prop:audit}. Therefore, there exists a unique solution 
\begin{equation*}
    X(t)=\phi (t-T_1,x_1), \qquad t \in [T_1,T_2),
\end{equation*}
of ODE \eqref{population_eq}, again due to Property \ref{prop0}. This procedure repeats and the audit PDMP
\begin{equation*}
    X(t)=\phi(t-T_i,x_i), \qquad t \in [T_i,T_{i+1}), \qquad i\geq 0,
\end{equation*}
where
\begin{equation*}
    x_i=X(T_i)=\Psi (X(T_i^-)) \in \mathcal{X},
\end{equation*}
is defined for all $t\geq 0$. Moreover, the process $X(t)$ has state space $\mathcal{X}$ and satisfies population conservation (Property \ref{prop1}) and global income conservation (Property \ref{prop2}).

Note also that, for any $x \in \mathcal{X}$, the map $x \mapsto \lambda_{\text{audit}}(x)$ \eqref{audit_rate} is measurable and continuous since it is linear in the state variables. Further, since the fractions satisfy $\sum_{j=1}^n\sum_{\alpha=1}^m x_j^\alpha=1$, 
a uniform bound is obtained
 \begin{equation*}
  0 \;\leq\; \lambda_{\text{audit}}(x) \;\leq\; \gamma, \qquad x \in \mathcal{X}.   
 \end{equation*}
The boundedness of the jump rate ensures non-explosion of the sequence of jump times $T_i$, i.e.
\begin{equation*}
\sup_{i\ge 1} T_i = \infty \quad \text{a.s}.    
\end{equation*}


\begin{remark}[A note on stationarity]
\label{Rem_St_aud}
Repeated audits systematically reduce the population fractions in the non-compliant sectors ($\alpha>1$) through the map $\Psi$ \eqref{audit_map}. This causes the process to concentrate in the compliant sector ($\alpha=1$) over time. As a result, the system is driven toward an equilibrium point $x^*$ dominated by compliance. 

Formally, this behaviour suggests two key properties, which are both supported by our numerical experiments in Section \ref{sec:audit_simulation}. First, for any initial state $x_0$ satisfying the conditions in Property \ref{prop0}, the distribution of $X(t)$ converges to a unique stationary distribution given by the Dirac measure $\boldsymbol{\delta}_{x^*}^{x_0}$, concentrated at a point $x^*$ dominated by compliance. Second, for different initial states $x_0^{(1)}, x_0^{(2)}, x_0^{(3)}, \ldots$ that, in addition to the conditions in Property \ref{prop0}, satisfy the conditions in Property \ref{prop3}, the corresponding stationary distributions $\boldsymbol{\delta}_{x^*}^{x_0^{(1)}}, \boldsymbol{\delta}_{x^*}^{x_0^{(2)}}, \boldsymbol{\delta}_{x^*}^{x_0^{(2)}}, \ldots$ are all concentrated at the same point $x^*$. This forms an analogue to Property \ref{prop3} in the stochastic PDMP case. 
\end{remark}

\subsection{Behavioural imitation}
\label{Sec_behave}

Behavioural imitation captures the tendency of individuals to adopt the tax evasion behaviour of their peers, reflecting the well-documented role of social norms in compliance. Individuals often act as conditional cooperators, basing their decisions on observed behaviour within their reference group. Empirical studies strongly support this idea. In fact, in \cite{alm2011ethics} the authors show that taxpayers who believe evasion is common are more likely to justify it themselves, and \cite{luttmer2014tax} finds positive correlations between the compliance of an individual and that of their social network. Moreover, formal economic models acknowledge that evasion can spread through social learning \cite{myles1996model}, and experiments confirm that observing non-compliance among similar agents increases one's own evasion \cite{fortin2007tax}.

We model imitation as a process that exclusively promotes higher evasion. Logically, individuals are unlikely to imitate behaviour that voluntarily increases their tax burden, but instead, learning about successful evasion strategies can undermine the perceived obligation to comply, making non-compliant behaviour contagious, see for example \cite{bicchieri2009right}. Hence, for simplicity, we assume imitation occurs only within the same income class, as individuals are primarily influenced by their socioeconomic peers \cite{bicchieri2009right}.

This mechanism is formalised as an imitation PDMP, acting as a counterpoint to the audit process. During an imitation event, a fraction of individuals in each sector transitions to the next higher evasion sector. This models the contagion of non-compliance, systematically reducing the population's overall compliance level. In this setting, an individual's income class remains unchanged during these jumps.

\paragraph{Mathematical formulation of the imitation PDMP.}

Similar to the audit PDMP, the system is modelled as a single-mode $n \times m$-dimensional PDMP
\begin{equation*}
    Z(t)=X(t)=(X_j^\alpha(t))_{j=1,\ldots,n;\alpha=1,\ldots,m}=\phi(t-T_i,x_i), \qquad t \in [T_i,T_{i+1}), \quad i\geq 0,
\end{equation*}
where $\phi$ is the flow of the ODE defined by \eqref{population_eq}, $T_0:=0$, and the $T_i$, $i\geq 1$, correspond to random imitation event times. This means that the process $X(t)$ again follows the $n\times m$-dimensional ODE \eqref{population_eq} piecewise on each inter-event interval $[T_i,T_{i+1})$, with a different initial condition $x_i=X(T_i)$ for each ``piece'', which reflects the  reallocation movements of population fractions caused by the imitation event happening at time $T_i$.

The random imitation events are also governed by an intensity function that depends on the current state of the population, particularly on the balance between compliant and non-compliant behaviour. A natural choice is
\begin{equation}\label{imitation_rate}
\lambda_{\text{imitation}}(x) = \eta \sum_{j=1}^{n} x_j^1,    
\end{equation}
where $\eta>0$ is a tunable parameter controlling the frequency of imitation events. Since this rate is proportional to the fraction of compliant individuals, imitation becomes more frequent when compliance is high, reflecting the idea that a larger compliant population generates more opportunities for transmitting information about evasion and thus adopting more evasive behaviour.

A constant imitation rate, $\lambda_{\text{imitation}}(x)\equiv \eta_0>0$, arises as the special case of the general specification in which the dependence on the compliant share is suppressed. This regime corresponds to spontaneous imitation events, e.g. socially driven behavioural shifts, that occur at random times according to a homogeneous Poisson process, independently of the current distribution of evasion levels.

At each imitation event time $T_i$, a fixed proportion $\epsilon\in(0,1]$ of individuals in every evasion sector is instantaneously reassigned to the next higher evasion sector, reflecting upward imitation towards more evasive behaviour. Formally, the imitation update is described by a map $\Phi$, defined component-wise by
\begin{equation}\label{imitation_map}
  \Phi^\alpha_j(x) =
  \begin{cases}
    (1-\epsilon)\,x^\alpha_j, & \alpha=1,\\[0.8ex]
    (1-\epsilon)\,x^\alpha_j + \epsilon\,x_j^{\alpha-1}, & 2\le \alpha \le m-1,\\[1.0ex]
    x^\alpha_j + \epsilon\,x_j^{\alpha-1}, & \alpha=m,
  \end{cases}
  \qquad j=1,\dots,n.
\end{equation}
Thus, for each income class $j$, a fraction $\epsilon x_j^\alpha$ of every sector $\alpha<m$ moves upward to sector $\alpha+1$, while the highest evasion sector $\alpha=m$ only receives inflow from $\alpha=m-1$. The fully compliant sector $\alpha=1$ loses a proportion $\epsilon x_j^1$ of its mass, reflecting the adoption of more evasive behaviour. The associated transition kernel is deterministic,
\begin{equation*}
   Q_{\text{imitation}}(x,\cdot) = \boldsymbol{\delta}_{\Phi(x)}(\cdot), 
\end{equation*}
and satisfies $Q_{\text{imitation}}(x,\{x\})=0$ whenever 
some class contains a positive fraction of compliant individuals or intermediate evaders. In this way, the imitation PDMP captures the systematic upward movement across evasion sectors induced by behavioural imitation.

\begin{remark}
    Stochastic extensions of the transition kernel are again possible, for instance by choosing~$\epsilon$ randomly in $(0,1]$ at each imitation event time.
\end{remark}

\paragraph{Well-posedness of the imitation PDMP.}

As in the audit setting, the imitation PDMP is well-defined and preserves the fundamental structural properties of the deterministic tax evasion model.

\begin{proposition}\label{prop:imitation}
For every $\epsilon \in (0,1]$ and $x\in \mathcal{X}$ \eqref{State_space}, the imitation update $\Phi$ \eqref{imitation_map} satisfies $\Phi(x) \in \mathcal{X}$.
\end{proposition}

\begin{proof}
Since $x \in \mathcal{X}$ and $\epsilon \in (0,1]$, by construction it holds that 
\begin{equation*}
    \Phi^\alpha_j(x) \geq 0, \quad \text{for all } j=1,\ldots,n \ \text{and} \ \alpha=1,\ldots,m.
\end{equation*}
Moreover, for each income class $j \in \{1,\ldots,n\}$, we have
\begin{align*}
\sum_{\alpha=1}^m \Phi^\alpha_j(x)
&= (1-\epsilon)x^1_j + \sum_{\alpha=2}^{m-1} \left[(1-\epsilon)x^\alpha_j + \epsilon x^{\alpha-1}_j\right] + (x^m_j + \epsilon x^{m-1}_j) \\
&=x_j^1-\epsilon x_j^1 + \sum_{\alpha=2}^{m-1} x_j^\alpha - \epsilon \sum_{\alpha=2}^{m-1}x_j^\alpha + \epsilon \sum_{\alpha=2}^{m-1} x_j^{\alpha-1}+ x_j^m +\epsilon x_j^{m-1} \\
&= \sum_{\alpha=1}^m x^\alpha_j - \epsilon \sum_{\alpha=1}^{m-1} x^\alpha_j + \epsilon \underbrace{\sum_{\alpha=2}^{m} x^{\alpha-1}_j}_{=\sum\limits_{\alpha=1}^{m-1} x^{\alpha}_j} 
= \sum_{\alpha=1}^m x^\alpha_j.
\end{align*}
Therefore, since $x\in \mathcal{X}$, and thus $\sum_{j=1}^n\sum_{\alpha=1}^m x_j^\alpha=1$, we also have that $\sum_{j=1}^n\sum_{\alpha=1}^m \Phi_j^\alpha(x)=1$, which completes the proof.
\end{proof}

Similar to the audit case, Proposition \ref{prop:imitation} readily implies that \begin{equation*}
\sum_{j=1}^n r_j \sum_{\alpha=1}^m \Phi_j^\alpha(x) = \sum_{j=1}^n r_j \sum_{\alpha=1}^m x_j^\alpha.
\end{equation*}
This is because $\Phi$ only reassigns individuals between evasion sectors within the same income class $j$, the total fraction $\sum_{\alpha=1}^m x^\alpha_j$ in each class remaining unchanged. Thus, both population conservation (cf. Property \ref{prop1}) and global income conservation (cf. Property~\ref{prop2}) are preserved by the imitation update map $\Phi$ \eqref{imitation_map}.

Analogously to the audit case, Proposition \ref{prop:imitation} implies that the imitation PDMP
\begin{equation*}
    X(t)=\phi(t-T_i,x_i), \qquad t \in [T_i,T_{i+1}), \qquad i\geq 0,
\end{equation*}
where
\begin{equation*}
    x_i=X(T_i)=\Phi (X(T_i^-)) \in \mathcal{X},
\end{equation*}
is defined for all $t\geq 0$, that it has state space $\mathcal{X}$, and that it satisfies population conservation (Property \ref{prop1}) and global income conservation (Property \ref{prop2}). In addition, the imitation jump rate \eqref{imitation_rate} is measurable and bounded, which guarantees non-explosion of the sequence of jump times~$T_i$.

\begin{remark}[A note on stationarity]\label{rem:imitation_stationary}
Repeated imitation events systematically push population mass toward higher evasion sectors through the map $\Phi$ \eqref{imitation_map}. Therefore, the system is driven toward a stationary state $x^*$, which is characterised by significantly higher levels of evasion compared to the deterministic model, with substantial mass concentrated in the highest evasion sector ($\alpha=m)$. 

Formally, these dynamics suggest the following two properties, which are supported by our numerical experiments in Section~\ref{Imit_sim}. First, for any initial state $x_0$ satisfying the conditions in Property \ref{prop0}, the distribution of $X(t)$ converges to a Dirac measure $\boldsymbol{\delta}_{x^*}^{x_0}$ concentrated at $x^*$. Second, for different initial states $x_0^{(1)},x_0^{(2)},x_0^{(3)},\ldots$ that, in addition to the conditions in Property~\ref{prop0}, satisfy the conditions in Property \ref{prop3}, the corresponding $\boldsymbol{\delta}_{x^*}^{x_0^{(1)}},\boldsymbol{\delta}_{x^*}^{x_0^{(2)}},\boldsymbol{\delta}_{x^*}^{x_0^{(3)}},\ldots$ are all concentrated at the same point $x^*$. This forms again an analogue to Property \ref{prop3} in the stochastic
PDMP case. 
\label{Imit_convg}
\end{remark}

\subsection{An illustrative example and model simulation}
\label{ex:3x2}

To illustrate the mechanics of the audit and imitation PDMPs, we recall the illustrative example introduced in Section \ref{Back_sec} with $n=3$  income classes and $m=2$ evasion sectors, and extend it to the corresponding PDMP framework. All simulations are based on a modified Euler scheme combined with a suitable thinning procedure (see Appendix~\ref{PDMP_simulation} for  details). Unless stated otherwise, the uniform initial condition
\begin{equation}
   x_0 = \begin{pmatrix}
       1/6 & 1/6 \\
       1/6 & 1/6 \\
       1/6 & 1/6 
   \end{pmatrix} 
   \label{initial_x_val}
\end{equation}
is used in the simulations.

\subsubsection{Numerical illustration: audit mechanism}
\label{sec:audit_simulation}

We start by numerically investigating the audit PDMP introduced in Section \ref{audit_sec}.

\paragraph{PDMP sample paths.}

We first examine the case of a constant jump rate $\lambda_{\text{audit}}(x) \equiv 1$, corresponding to audits modelled by a homogeneous Poisson process with unit intensity. With audit effectiveness $\delta = 0.1$ in \eqref{audit_map}, each audit transfers $10\%$ of partial evaders (sector $\alpha=2$) to the compliant sector ($\alpha=1$).
\begin{figure}
    \centering
    \includegraphics[width=0.8\textwidth]{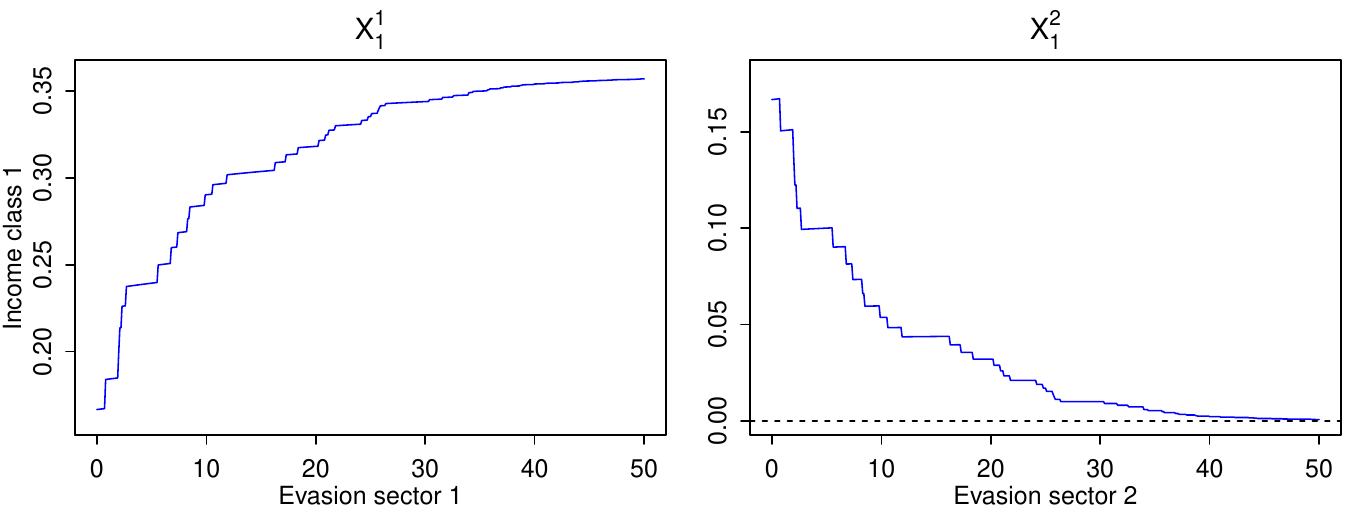}
    \caption{
        \textbf{Audit PDMP with constant jump rate.} Possible evolution of the population fractions in income class $j=1$ for the compliant 
        sector ($\alpha=1$, left panel) and the non-compliant sector 
        ($\alpha=2$, right panel). The audit effectiveness in \eqref{audit_map} is set to $\delta=0.1$.
    }
    \label{fig:constant_audit}
\end{figure}

The dynamics of this audit mechanism are illustrated in Figure~\ref{fig:constant_audit}. The compliant fraction increases at the jump times,  
while the non-compliant fraction decreases monotonically and asymptotically approaches zero over time. This systematic decay of the non-compliant population indicates that repeated audits drive the system toward a fully compliant state.

\begin{remark}
    Note that audit events are represented by kinks in the population fraction  trajectories. For visual clarity, the states immediately before and after each audit are connected by straight line segments, although these correspond to discontinuous transitions in the underlying PDMP. Between audits, the system follows the deterministic flow given by \eqref{population_eq}, resulting in the smooth segments of the trajectories.
\end{remark}

Now, we consider the more realistic case of a state-dependent audit rate defined by \eqref{audit_rate}, which reduces to
\begin{equation}\label{audit_rate2}
\lambda_{\text{audit}}(x) = \gamma \sum_{j=1}^{3} x_j^2,    
\end{equation}
for the illustrative example. This rate is proportional to the total non-compliant population, reflecting the typical allocation of audit resources toward sectors with higher evasion. 

\begin{figure}
    \centering
    \includegraphics[width=1.0\textwidth]{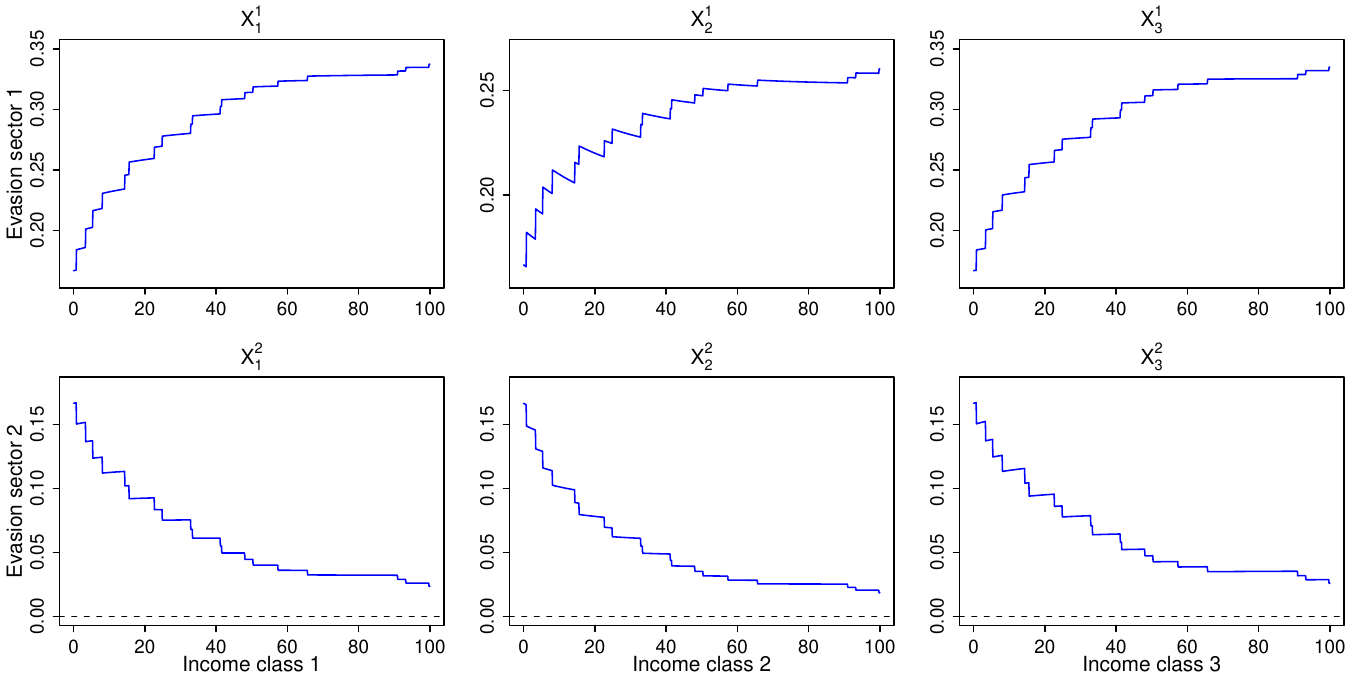}
    \caption{%
        \textbf{Audit PDMP with state-dependent jump rate.} Possible evolution of population fractions. The audit intensity in \eqref{audit_rate2} is set to $\gamma=1$, and the audit effectiveness in \eqref{audit_map} is set to $\delta=0.1$. 
    }
    \label{fig:state_dependent_audit}
\end{figure}

Figure~\ref{fig:state_dependent_audit} highlights qualitative differences from the constant-rate case shown in Figure~\ref{fig:constant_audit}. The most notable distinction is the reduced frequency of audit events, which arises from the state-dependent jump rate \eqref{audit_rate2}.  
As the non-compliant population fractions decrease over time, audits become less frequent. This creates a feedback mechanism where successful audits reduce the likelihood of future audit events. 
\paragraph{Long-time behaviour.}

To empirically investigate the long-time behaviour suggested in Remark~\ref{Rem_St_aud}, we examine paths of the audit PDMP obtained under the three different initial conditions \eqref{initial_conditions}. Different paths obtained under the same $x_0$, tend to the same equilibrium point $x^*$, so that one representative path for each $x_0$ is reported in Figure~\ref{fig:multi_initial}. 
Qualitatively, all three paths exhibit similar long-term behaviour: the population fractions of the non-compliant sector ($\alpha=2$, bottom panels) decrease over time, while those of the compliant sector ($\alpha=1$, top panels) become dominant. The apparent convergence of trajectories from different initial conditions $x_0$ provides preliminary numerical evidence for the existence of stationary distributions corresponding to Dirac measures $\boldsymbol{\delta}_{x^*}^{x_0}$ concentrated at the same point $x^*$, which is characterised by full compliance. 

\begin{figure}[h!]
    \centering
    \includegraphics[width=1.0\textwidth]{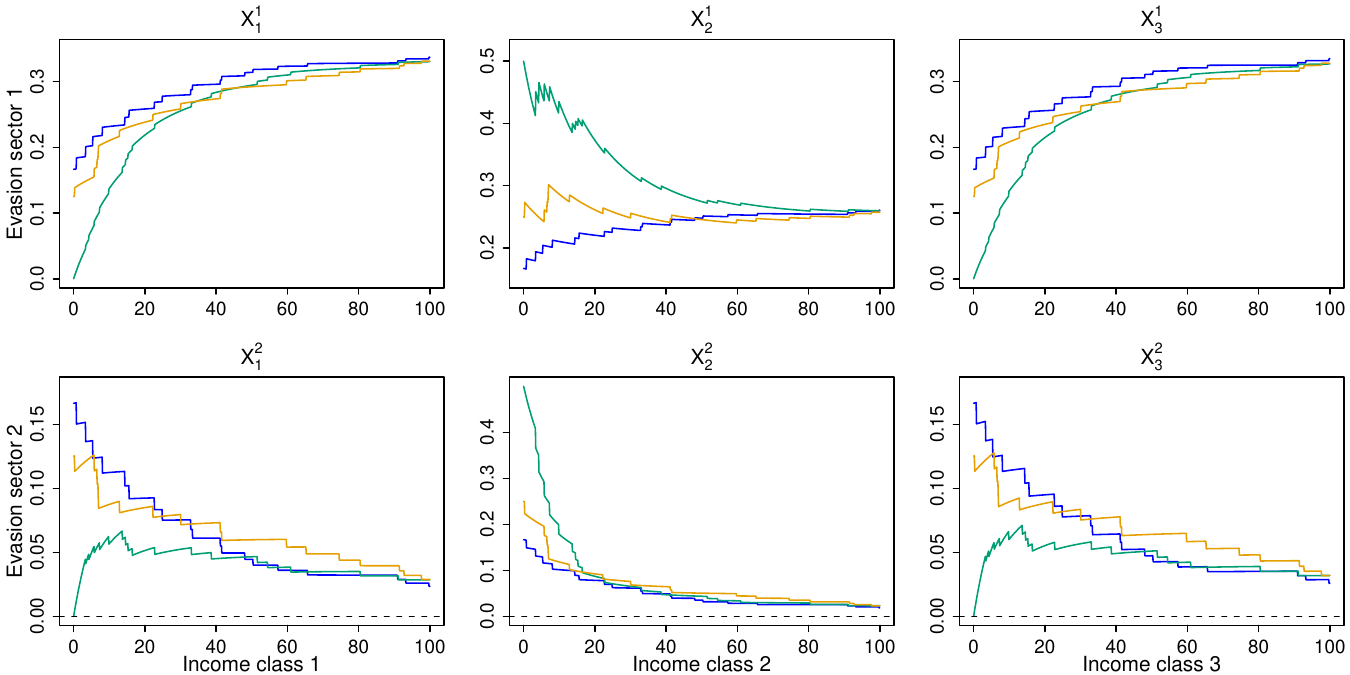}
    \caption{\textbf{Population fraction paths of the audit PDMP} based on the three initial conditions~\eqref{initial_conditions}. The audit intensity in \eqref{audit_rate2} is set to $\gamma=1$, and the audit effectiveness in \eqref{audit_map} is set to $\delta=0.1$.}
    \label{fig:multi_initial}
\end{figure}

\paragraph{Audit intensity and effectiveness.}

The audit mechanism depends critically on two policy parameters: the audit intensity $\gamma$ and the audit effectiveness $\delta$. We investigate their influence through systematic variation of these parameters.

\begin{figure}[h!]
    \centering
    \includegraphics[width=0.8\textwidth]{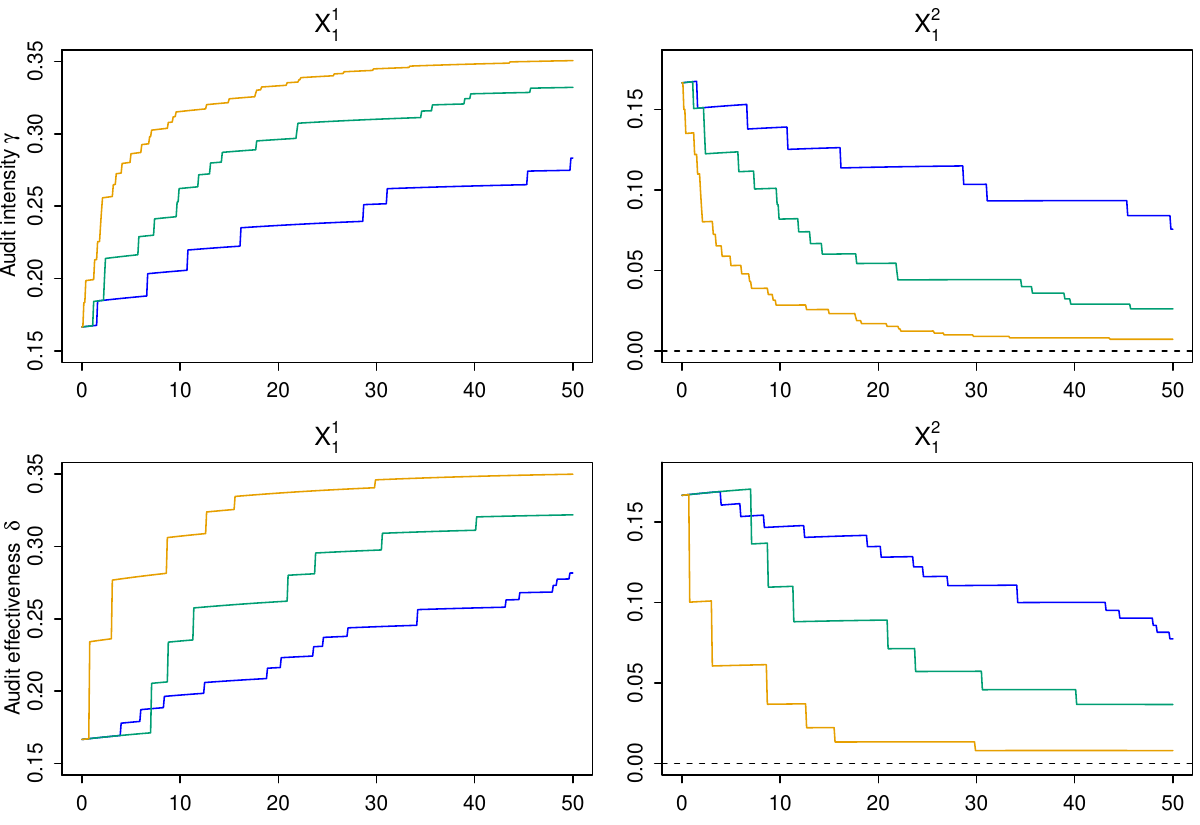}
    \caption{%
        \textbf{Effect of audit intensity $\gamma$} (top panels). Population fraction paths of income class~$1$ for different audit intensity values (blue: $\gamma=0.5$, green: $\gamma=2$, orange: $\gamma=10$). \textbf{Effect of audit effectiveness $\delta$} (bottom panels). Population fraction paths of income class $1$ for different audit effectiveness values (blue: $\delta=0.05$, green: $\delta=0.2$, orange: $\delta=0.4$).
    }
    \label{fig:gamma_sensitivity}
\end{figure}

Figure~\ref{fig:gamma_sensitivity} (top panels) shows that the audit intensity $\gamma$ regulates the temporal frequency of enforcement interventions. Higher values of $\gamma$ lead to more frequent audit events, accelerating the reallocation of population fractions toward compliance, consistent with the state-dependent jump rate \eqref{audit_rate2}.

Figure~\ref{fig:gamma_sensitivity} (bottom panels) illustrates how the audit effectiveness parameter $\delta$ influences the impact of each audit on the population. Larger values of $\delta$ lead to bigger instantaneous transfers from the non-compliant to the compliant sector during audit events, visible as larger steps in the trajectories. Thus, $\delta$ controls the efficiency of individual audit events in reallocating population toward compliance. 

This analysis suggests that both parameters independently influence compliance behaviour. The parameter $\gamma$ determines how frequently audits occur, while $\delta$ determines how effective each audit is in reducing evasion behaviour. Both parameters can achieve similar long-term compliance levels, but with different temporal patterns and potentially different economic costs.

\subsubsection{Numerical illustration: imitation mechanism}
\label{Imit_sim}

We now numerically examine the imitation PDMP introduced in Section \ref{Sec_behave}.

\paragraph{PDMP sample paths.}

We first consider a constant jump rate $\lambda_{\text{imitation}}(x) \equiv 1$, so that imitation events occur again according to a homogeneous Poisson process with unit intensity. The imitation strength parameter in \eqref{imitation_map} is set to $\epsilon = 0.1$. 

\begin{figure}[h!]
    \centering
    \includegraphics[width=0.8\textwidth]{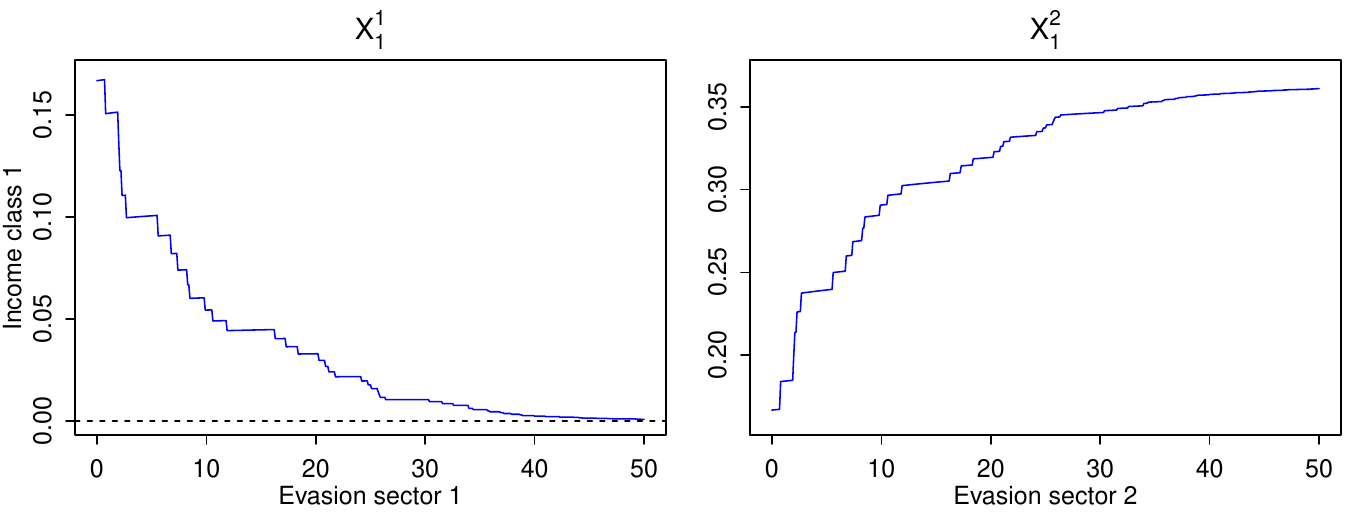}
    \caption{\textbf{Imitation PDMP with constant jump rate.} Possible evolution of the population fractions in income class $j=1$ for the compliant 
        sector ($\alpha=1$, left panel) and the non-compliant sector 
        ($\alpha=2$, right panel). The imitation strength parameter in \eqref{imitation_map} is set to $\epsilon = 0.1$.}
    \label{fig:imitation_constant}
\end{figure}

Figure~\ref{fig:imitation_constant} shows a typical trajectory of the imitation PDMP. The imitation events are represented by abrupt downward jumps in the compliant sector (left panel) accompanied by abrupt upward jumps of equal magnitude in the non-compliant sector (right panel), consistent with population conservation (cf. Property \ref{prop1}). The population fraction in the compliant group thus decreases at the imitation events and approaches zero as time evolves, while the non-compliant sector stabilises at a certain value. This reflects the tendency of the imitation dynamics to drive the system toward a fully non-compliant state. 

We now consider the more realistic state-dependent imitation jump rate defined by \eqref{imitation_rate}, which reads as
\begin{equation}
\label{eq:lambda-imitation}
\lambda_{\text{imitation}}(x) = \eta \sum_{j=1}^3 x_j^1,
\end{equation}
for the illustrative example. 

\begin{figure}
    \centering
    \includegraphics[width=1.0\textwidth]{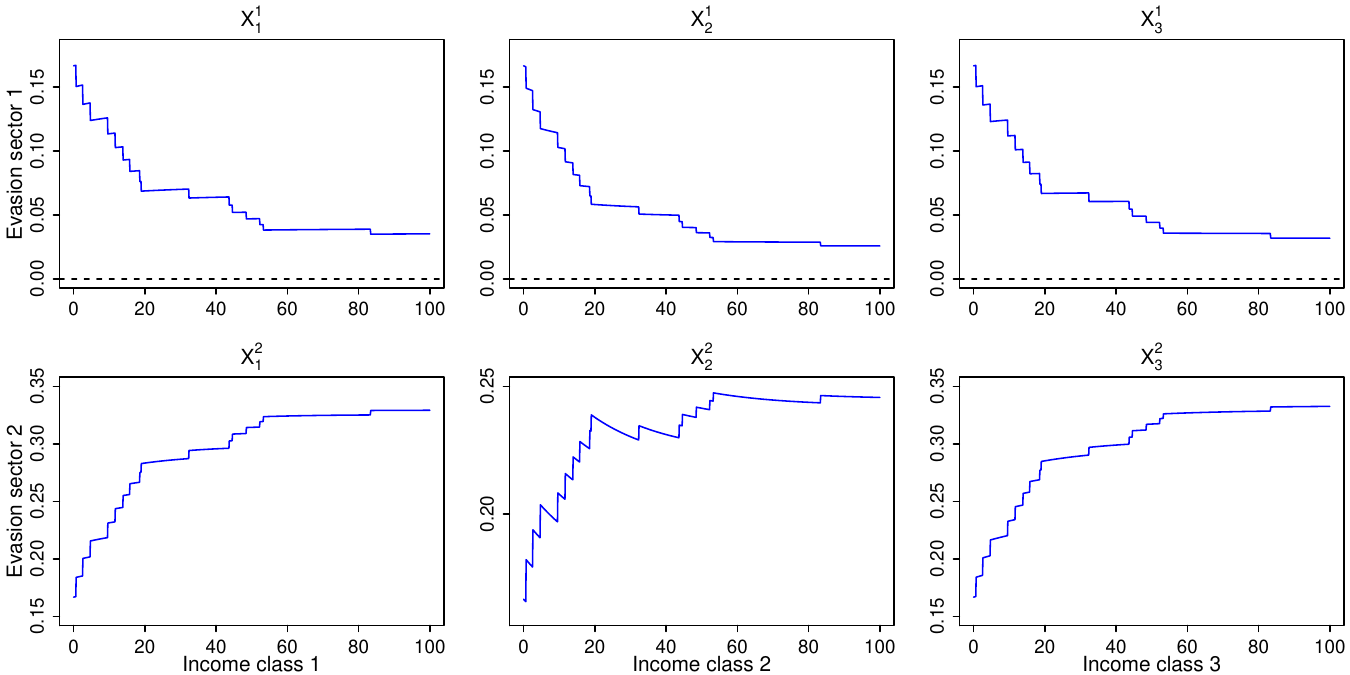}
    \caption{\textbf{Imitation PDMP with state-dependent jump rate.} Possible evolution of population fractions. The imitation frequency in \eqref{eq:lambda-imitation} is set to $\eta=1$, and the imitation strength in \eqref{imitation_map} is set to $\epsilon=0.1$.}
    \label{fig:imitation_state_dependent}
\end{figure}

Figure~\ref{fig:imitation_state_dependent} shows significantly fewer imitation events compared to the constant-rate case in Figure~\ref{fig:imitation_constant}. This reduction stems from the state-dependent nature of the jump rate \eqref{eq:lambda-imitation}. As the compliant population fractions decrease over time, imitation events become less frequent, creating a self-limiting effect.  

\paragraph{Long-time behaviour.}

\begin{figure}[h!]
    \centering
    \includegraphics[width=1.0\textwidth]{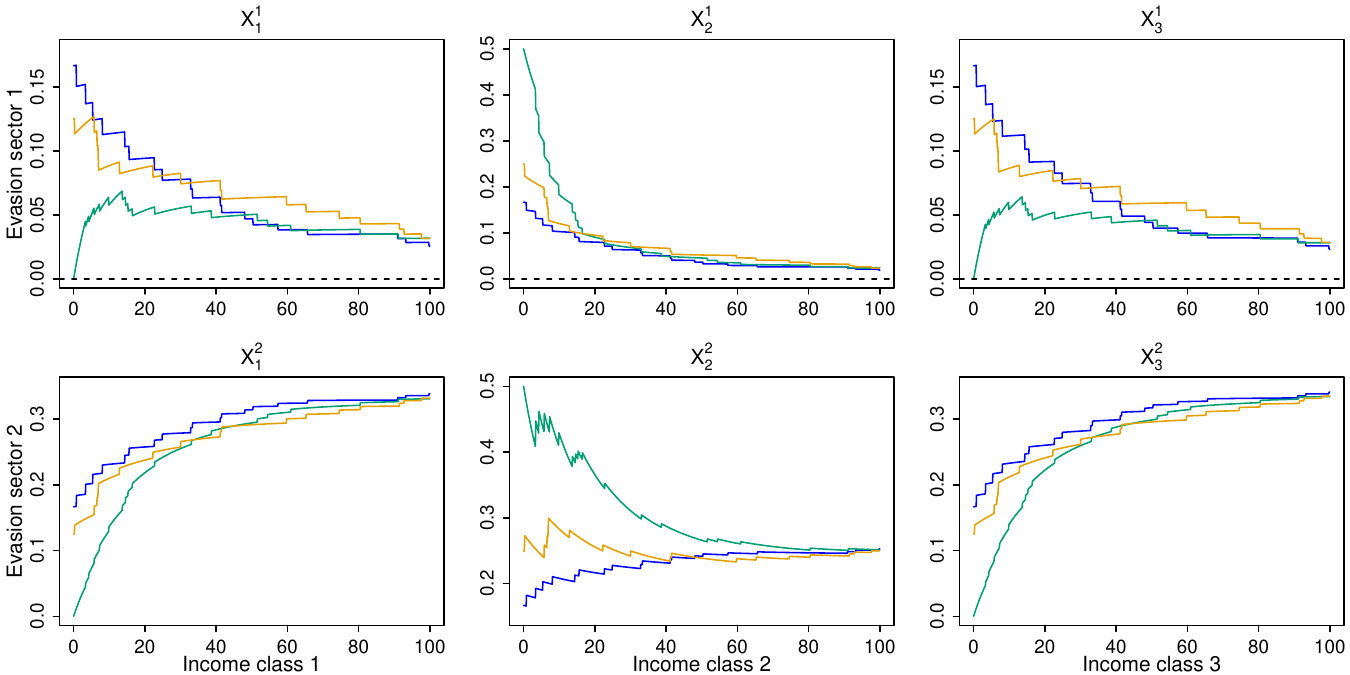}
    \caption{\textbf{Population fraction paths of the imitation PDMP} based on the three initial conditions \eqref{initial_conditions}. The imitation frequency in \eqref{eq:lambda-imitation} is set to $\eta=1$, and the imitation strength in \eqref{imitation_map} is set to $\epsilon=0.1$.}
    \label{fig:imitation_initial_conditions}
\end{figure}

To empirically assess the long-term behaviour suggested in Remark~\ref{rem:imitation_stationary}, we simulate trajectories of the imitation PDMP from the three initial states \eqref{initial_conditions} and report them in Figure~\ref{fig:imitation_initial_conditions}. Different paths obtained under the same $x_0$ tend to the same equilibrium point $x^*$, so that only one representative path for each $x_0$ is reported to ease figure readability. For all three trajectories, the population fractions of the compliant sector ($\alpha=1$, top panels) decrease over time, while those of the compliant sector ($\alpha=2$, bottom panels) become dominant. Specifically, all trajectories seem to converge to the same equilibrium $x^*$ from different initial states $x_0$, indicating the existence of underlying stationary distributions corresponding to  Dirac measures $\boldsymbol{\delta}_{x^*}^{x_0}$ concentrated at the same point $x^*$, which is characterised by full non-compliance. 

The audit and imitation mechanisms cause fundamentally different long-term behaviours. While audit dynamics (see Figure~\ref{fig:multi_initial}) drive the system toward full compliance through systematic reduction of non-compliant population fractions, imitation dynamics (see Figure \ref{fig:imitation_initial_conditions}) produce the opposite effect, driving the system toward full non-compliance. These opposing effects arise due to the directional nature of each mechanism.

\paragraph{Imitation frequency and strength.}

\begin{figure}
    \centering
    \includegraphics[width=0.8\textwidth]{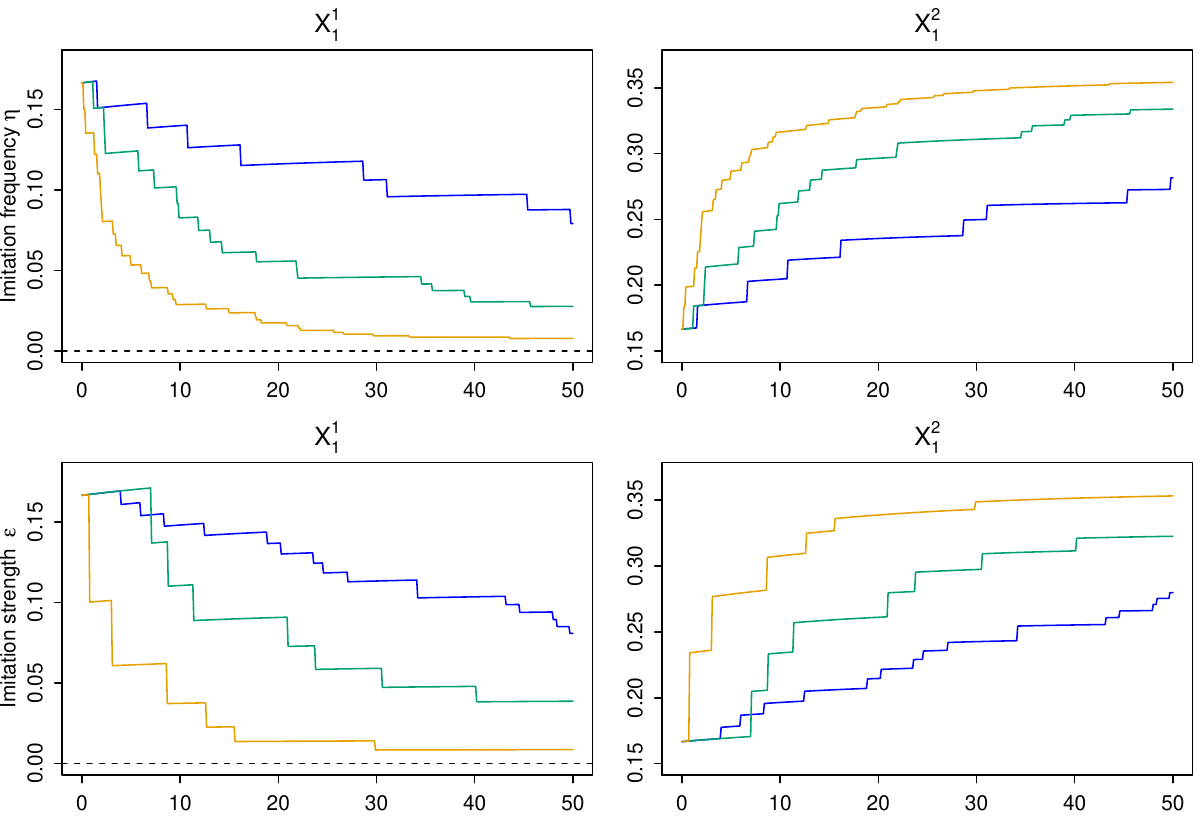}
    \caption{\textbf{Effect of imitation frequency $\eta$} (top panels). Population fraction paths of income class $1$ for different imitation frequency values (blue: $\eta=0.5$, green: $\eta=2$, orange: $\eta=10$). \textbf{Effect of imitation strength $\epsilon$} (bottom panels). Population fraction paths of income class $1$ for different imitation strength values (blue: $\epsilon=0.05$, green: $\epsilon=0.2$, orange: $\epsilon=0.4$).} 
    \label{fig:imitation_eta}
\end{figure}

The imitation mechanism depends critically on the imitation frequency parameter $\eta$ and the imitation strength parameter $\epsilon$. We investigate their impact on the population dynamics by systematically varying their values.

As shown in Figure~\ref{fig:imitation_eta} (top panels), higher values of $\eta$ increase the frequency of imitation events determined by the jump rate \eqref{eq:lambda-imitation}. Thus, the larger $\eta$, the faster is the reallocation of population fractions toward non-compliance. 

Moreover, Figure~\ref{fig:imitation_eta} (bottom panels) shows that larger values of $\epsilon$ produce more substantial jumps, as a greater fraction of the compliant population transitions to non-compliance at each imitation event. The parameter $\epsilon$ thus controls the magnitude of population transfers during jumps, while leaving the jump frequency unchanged. 

Both parameters accelerate convergence toward full non-compliance, but through distinct mechanisms: $\eta$ increases the rate of imitation events, while $\epsilon$ amplifies the population transfer per event. Consequently, increasing either parameter reduces the time required for the system to stabilise at a non-compliant equilibrium.

Recall that, in audit dynamics, both the intensity $\gamma$ and effectiveness $\delta$ accelerate convergence toward compliance, while in imitation dynamics, the frequency $\eta$ and strength $\epsilon$ accelerate convergence toward non-compliance. This fundamental asymmetry reflects the competing influences of enforcement policy versus social imitation effects in shaping taxpayer behaviour. 


\section{A PDMP model for combined audit and imitation dynamics}
\label{Combined_mod}

We now integrate the two behavioural mechanisms, audits and imitation, into a unified PDMP framework. The resulting dynamics combine deterministic kinetic evolution with stochastic transitions between compliance sectors, where each jump corresponds to either an audit or an imitation event. Unlike the individual mechanisms that drive the system toward extreme states (full compliance for audits and full non-compliance for imitation), the combined model permits intermediate equilibrium distributions that may better reflect realistic compliance patterns.

\subsection{Mathematical formulation}\label{sec:combined_PDMP}

Let $X(t) = (X_j^\alpha(t))_{j=1,\dots,n;\,\alpha=1,\dots,m}$ represent the population state at time $t$, where $X_j^\alpha(t)$ describes the fraction of individuals in income class $j$ and evasion sector $\alpha$. The process is characterised by the PDMP characteristic triple $(\phi, \Lambda, Q)$, detailed as follows. 

On each inter-event interval $[T_i,T_{i+1})$, the system follows the ODE dynamics \eqref{population_eq} with flow~$\phi$ and initial condition $x_i=X(T_i)$, i.e.
\begin{equation}\label{combined_flow}
    X(t)=\phi(t-T_i,x_i), \qquad t \in [T_i,T_{i+1}), \quad i\geq 0.
\end{equation}

The event times $T_i$ are governed by a state-dependent rate function given by
\begin{equation}
\label{eq:combined-rate}
    \Lambda(x) = \lambda_{\text{audit}}(x) + \lambda_{\text{imitation}}(x),
\end{equation}
where $\lambda_{\text{audit}}(x)$ and $\lambda_{\text{imitation}}(x)$ are as in \eqref{audit_rate} and \eqref{imitation_rate}, respectively.
This formulation captures the natural dependence; audit likelihood increases with non-compliant population, while imitation likelihood increases with compliant population. At each jump time, the event type is determined probabilistically as follows
\begin{equation*}
\mathbb{P}(\text{audit} \mid \text{jump at } t) = \frac{\lambda_{\text{audit}}(X(t))}{\Lambda(X(t))}, \quad 
\mathbb{P}(\text{imitation} \mid \text{jump at } t) = \frac{\lambda_{\text{imitation}}(X(t))}{\Lambda(X(t))}.
\end{equation*}
The constant rate $\Lambda(x)\equiv \gamma_0+\eta_0>0$ is considered as a special case of \eqref{eq:combined-rate}.

Moreover, the post-jump states are determined by the selected event type through the transition kernel
\begin{equation}\label{combined_kernel}
Q(x, \cdot) = 
\begin{cases}
\boldsymbol{\delta}_{\Psi(x)}, & \text{with probability } \lambda_{\mathrm{audit}}(x)/\Lambda(x), \\
\boldsymbol{\delta}_{\Phi(x)}, & \text{with probability } \lambda_{\mathrm{imitation}}(x)/\Lambda(x),
\end{cases}    
\end{equation}
where $\delta_y$ denotes the Dirac measure at $y$, and $\Psi$ and $\Phi$ are the audit and imitation transition maps defined in \eqref{audit_map} and \eqref{imitation_map}, respectively.

The combined PDMP $(\phi, \Lambda, Q)$ generates rich dynamical behaviour through the interaction of opposing forces: audits promote compliance while imitation encourages non-compliance. The resulting dynamics depend critically on the parameter pairs $(\gamma, \eta)$ and $(\delta, \epsilon)$, which are summarised in Table~\ref{tab_params}. 

\begin{table}
\centering
\begin{tabular}{clp{8.5cm}}
\toprule
\textbf{Parameter} & \textbf{Mechanism} & \textbf{Interpretation} \\
\midrule
$\gamma$ & Audit intensity &
Rate of audit events; increases with the fraction of non-compliant individuals. \\[0.6ex]

$\delta$ & Audit effectiveness &
Share of non-compliant individuals that become compliant after an audit. \\[0.6ex]

$\eta$ & Imitation frequency &
Rate of imitation events; increases with the fraction of compliant individuals. \\[0.6ex]

$\epsilon$ & Imitation strength &
Share of individuals shifting one sector upward at each imitation event. \\
\bottomrule
\end{tabular}

\caption{Parameters of the combined audit–imitation PDMP.}
\label{tab_params}
\end{table}

\begin{remark}[Well-posedness]
    Note that the definition of the transition kernel \eqref{combined_kernel}, together with Propositions \ref{prop:audit} (audit) and \ref{prop:imitation} (imitation), implies that population conservation (cf. Property~\ref{prop1}) and global income conservation (cf. Property~\ref{prop2}) to be preserved at the transition times $T_i$. Therefore, all ``initial conditions'' $x_i=X(T_i)$ (cf. \eqref{combined_flow}) are in $\mathcal{X}$ (cf. \eqref{State_space}). This allows to use Property \ref{prop0}, which implies that there exists a unique solution to ODE \eqref{population_eq} on each inter-event interval $[T_i,T_{i+1})$. This solution, in turn, has state space $\mathcal{X}$ and conserves population and global income, i.e. it satisfies Properties \ref{prop1} and \ref{prop2}.
\end{remark}

\begin{remark}[Stationarity]
    An analogue to Property \ref{prop3} is suggested by our numerical experiments in Section \ref{sec:num_exp_combinedPDMP} for the illustrative example setting. We observe that, for any initial state $x_0$ satisfying the conditions in Property~\ref{prop0}, the distribution of $X(t)$ converges to a unique stationary distribution~$\mathcal{D}^{x_0}$.  
    In addition, for different initial states $x_0^{(1)},x_0^{(2)},x_0^{(3)},\ldots$ satisfying the conditions in Properties \ref{prop0} and \ref{prop3}, the corresponding stationary distributions $\mathcal{D}^{x_0^{(1)}},\mathcal{D}^{x_0^{(2)}},\mathcal{D}^{x_0^{(3)}},\ldots$ are all the same. Moreover, for $\gamma=\eta$ and $\delta=\epsilon$ (i.e., when audit and imitation events yield balancing effects), this common stationary distribution is centred near the stationary state $x^*$ of the underlying ODE.
\end{remark}

\subsection{Numerical illustration}
\label{sec:num_exp_combinedPDMP}

In this section, we numerically investigate the combined PDMP presented in Section~\ref{sec:combined_PDMP}. In particular, we consider again the illustrative example setting with $n=3$ income classes and $m=2$ evasion sectors, and adapt it to the combined PDMP framework. Unless stated otherwise, the initial condition  \eqref{initial_x_val} is again used.

\paragraph{PDMP sample paths.}

We first consider the case of a constant jump rate $\Lambda(x)\equiv 1$ and fixed audit and imitation probabilities $\mathbb{P}(\text{audit} \mid \text{jump at } t)=\mathbb{P}(\text{imitation} \mid \text{jump at } t)=1/2$, so that audit and imitation events occur equally often on
average. The audit effectiveness and imitation strength parameters are set to $\delta=\epsilon=0.1$. 

\begin{figure}[h!]
    \centering
    \includegraphics[width=0.8\textwidth]{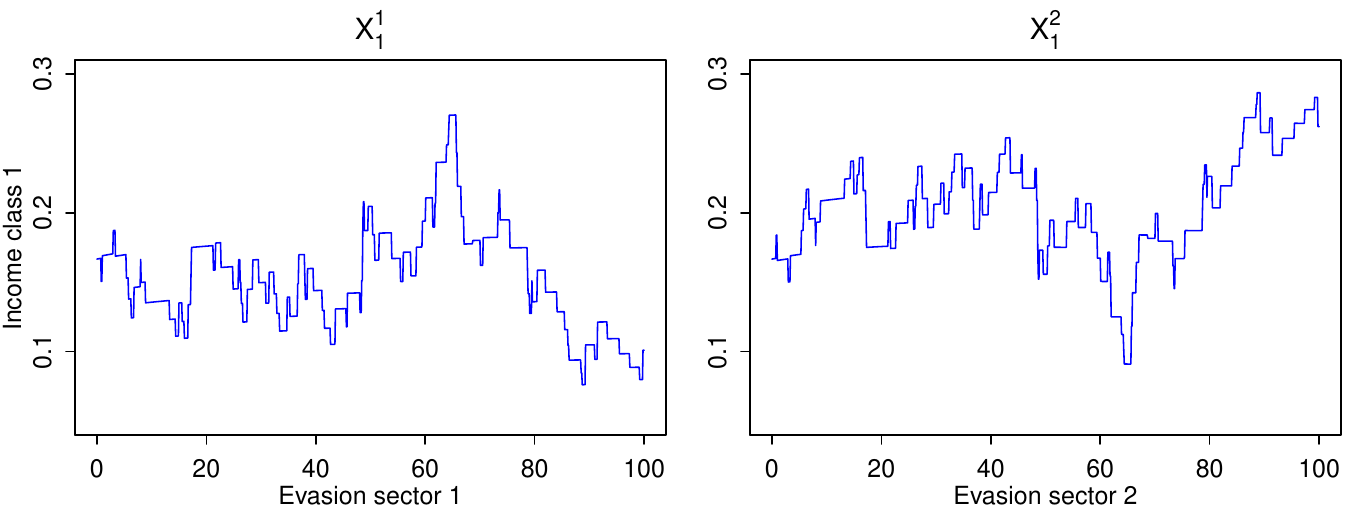}
    \caption{\textbf{Combined PDMP with constant jump rate.} 
        Possible evolution of the population fractions in income class $j=1$ for the compliant sector 
        ($\alpha=1$, left panel) and the non-compliant sector ($\alpha=2$, right panel). The audit effectiveness in \eqref{audit_map} and imitation strength in \eqref{imitation_map} are set to $\delta=\epsilon=0.1$.}
    \label{fig:combined_constant}
\end{figure}

Figure~\ref{fig:combined_constant} shows a possible trajectory of the combined PDMP for the first income class. 
The system now exhibits bidirectional movements. In particular, abrupt fraction reductions in the compliant sector correspond to imitation events, while fraction gains correspond to audit interventions. Similarly, fraction reductions in the non-compliant sector are du to audits, and fraction gains due to imitation events. In contrast to the 
pure audit or imitation models, neither sector dominates in the long term. Instead, since audit and imitation events occur with equal probability and possess symmetric jump magnitudes, the trajectory of the combined PDMP fluctuates around the solution of the underlying ODE model.\\

We now examine the combined PDMP under the full state-dependent jump rate \eqref{eq:combined-rate} with parameters $\gamma = 1$ (audit intensity) and $\eta = 2$ (imitation frequency).
The audit effectiveness and imitation strength are again set to $\delta=\epsilon=0.1$.

\begin{figure}
    \centering
    \includegraphics[width=1.0\textwidth]{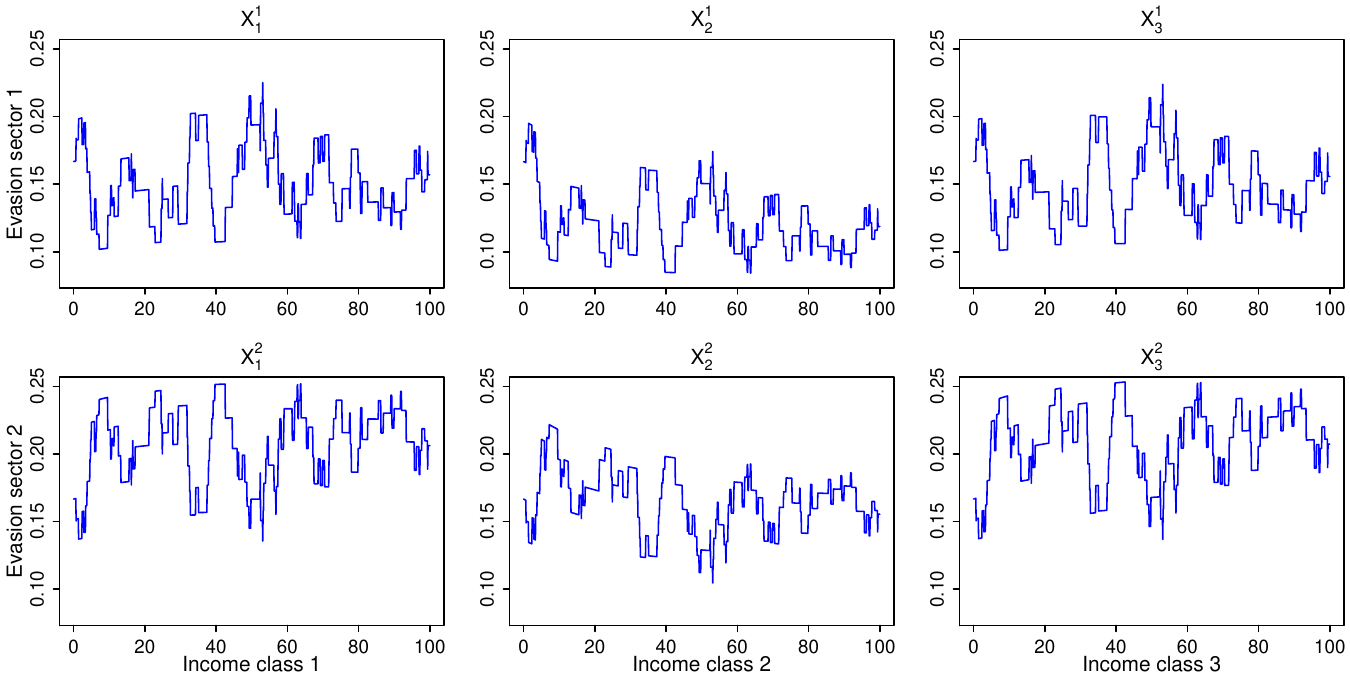}
    \caption{\textbf{Combined PDMP with state-dependent jump rate.} 
        Possible evolution of population fractions. The intensity parameters are $\gamma=1$ (audit) and $\eta=2$ (imitation), and the effectiveness parameters are $\delta=0.1$ (audit) and $\epsilon=0.1$ (imitation).}
    \label{fig:combined_state_dependent}
\end{figure}

The trajectory shown in Figure~\ref{fig:combined_state_dependent} exhibits increased jump frequency compared to the constant-rate case, due to the higher imitation rate parameter $\eta = 2$. In contrast to the pure audit or pure imitation model, where the system approaches a fully compliant or non-compliant state, in the combined model reallocation movements persist throughout the evolution of the system.

\paragraph{Long-time behaviour.}

\begin{figure}[h!]
    \centering
    \includegraphics[width=1.0\textwidth]{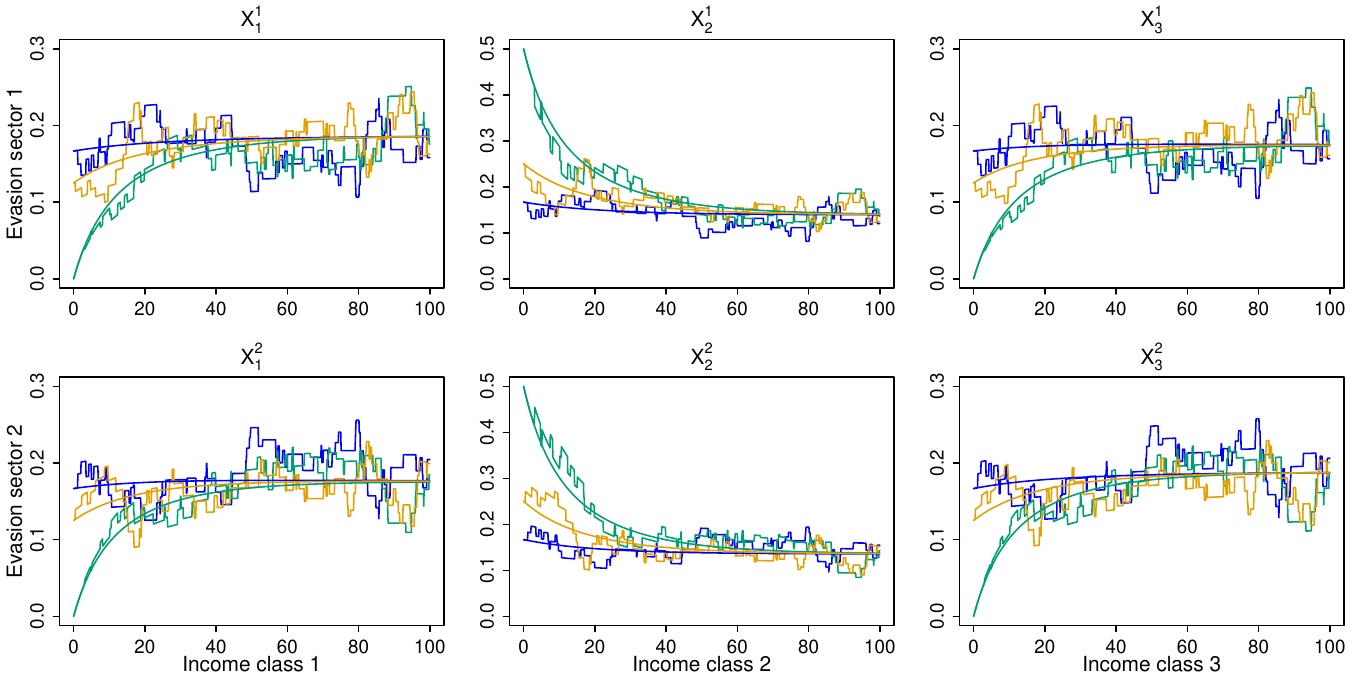}
    \caption{\textbf{Population fraction paths of the combined PDMP} based on the three initial conditions \eqref{initial_conditions}. The intensity parameters are $\gamma=1$ (audit) and $\eta=1$ (imitation), and the effectiveness parameters are $\delta=0.1$ (audit) and $\epsilon=0.1$ (imitation). The corresponding underlying ODE solutions are shown for reference (cf. Figure~\ref{fig:ODE_example_setting}).}
    \label{fig:combined_initial_conditions}
\end{figure}

To empirically investigate the long-time behaviour of the combined PDMP, we consider sample paths (for $\gamma=\eta=1$ and $\delta=\epsilon=0.1$), obtained under the three different initial conditions \eqref{initial_conditions}. Figure~\ref{fig:combined_initial_conditions} shows these paths, in comparison with the underlying ODE solutions (cf. Figure \ref{fig:ODE_example_setting}). Although the combined PDMP system does not converge to a single stationary state $x^*$ (or a dirac measure concentrated at $x^*$), the trajectories seem to approach a common region  concentrated near the common stationary state $x^*$ of the underlying ODE solutions. Note that different paths obtained under the same $x_0$ tend to the same region as time evolves, and thus only one representative path for each $x_0$ is reported to ease figure interpretability. These observations suggest the existence of stationary distributions~$\mathcal{D}^{x_0}$ for the combined PDMP model, which all coincide for different $x_0$ satisfying the conditions in Properties~\ref{prop0} and \ref{prop3}.

\begin{figure}
    \centering
    \includegraphics[width=1.0\textwidth]{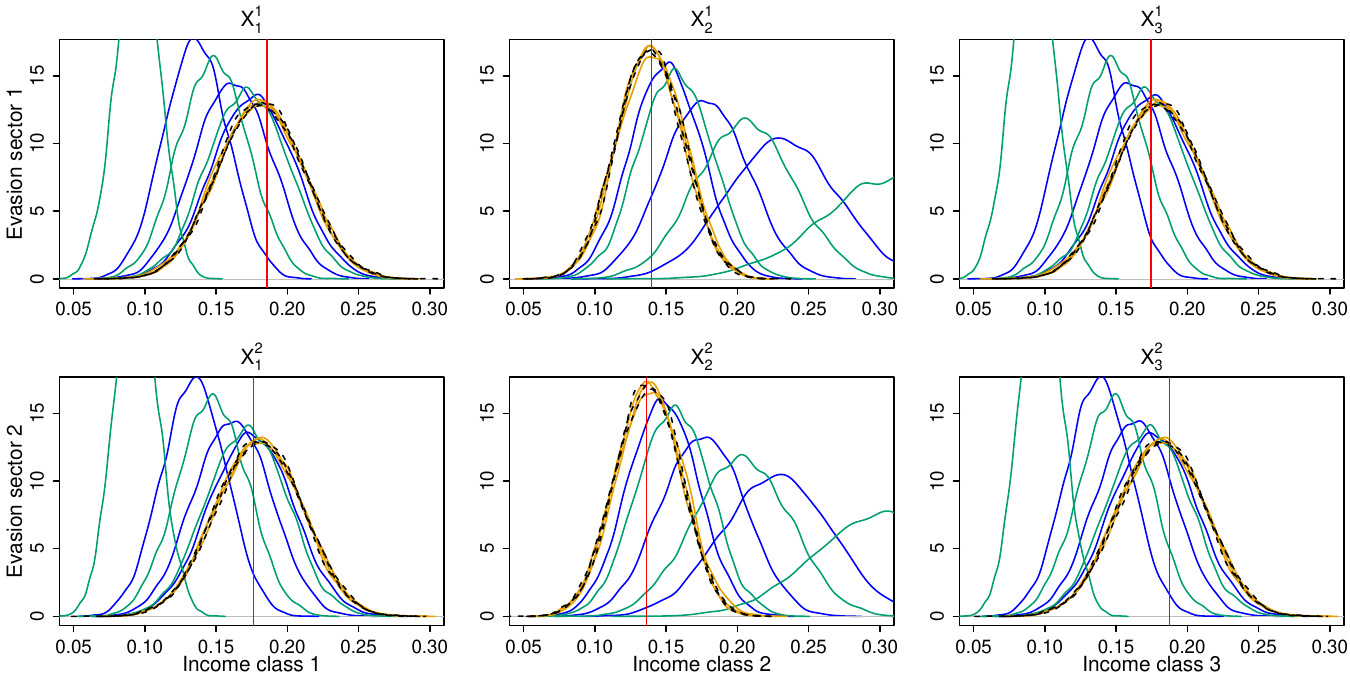}
    \caption{\textbf{Stationary distribution of the combined PDMP.} Empirical densities of population fractions at times $t=10$ (green), $t=20$ (blue), $t=90$ (orange), and $t=100$ (dashed black), computed from $10^4$ paths for each of the initial conditions $x_0^{(1)}$, $x_0^{(2)}$, and $x_0^{(3)}$ in \eqref{initial_conditions}. The vertical red lines indicate the stationary state $x^*$ of the underlying ODE. The intensity parameters are~$\gamma=1$ (audit) and $\eta=1$ (imitation), and the effectiveness parameters are $\delta=0.1$ (audit) and $\epsilon=0.1$ (imitation).
    }
    \label{fig:stationary_distribution}
\end{figure}

This is confirmed by Figure~\ref{fig:stationary_distribution}, where we report the empirical densities of the population fraction processes $X_j^\alpha(t)$ at different points in time, namely at $t=10$ (green), $t=20$ (blue), $t=90$ (orange), and $t=100$ (dashed black), computed from $10^4$ sample paths of the combined PDMP model for each of the initial conditions $x_0^{(1)}$, $x_0^{(2)}$, and $x_0^{(3)}$ in \eqref{initial_conditions}. For each $x_0$, the empirical densities differ initially (for smaller values of $t$; see the green and blue curves), but they overlap for sufficiently large times $t$ (see the orange and dashed black curves). Moreover, the densities at $t=90$ (orange) and $t=100$ (dashed black) all overlap for the different initial conditions. This provides strong numerical evidence that the system converges to the same stationary distribution from differential initial states satisfying the conditions of Properties \ref{prop0} and \ref{prop3}.

Note that, this common stationary distribution, for the setting considered, is centred near the common stationary state~$x^*$ of the underlying ODE model, which is indicated by the vertical red lines. 
However, a systematic deviation of the stationary mean from $x^*$ is visible. Specifically, for the lowest income class $j=1$, the stationary mean lies to the left of $x^*$ in the compliant sector ($\alpha=1$) and to the right of $x^*$ in the non-compliant sector ($\alpha=2$); for the highest income class $j=3$, the opposite pattern occurs (right in compliant, left in non-compliant); while for the middle income class $j=2$, the stationary mean is nearly aligned with $x^*$ in both sectors. These shifts persist even when the audit and imitation intensities are chosen symmetrically ($\gamma = \eta = 1$), ruling out a simple explanation based on unequal jump frequencies. Instead, the deviation arises from the interaction between the transition kernels $\Psi$ and $\Phi$ (defined in \eqref{audit_map} and \eqref{imitation_map}) and the nonlinear, class-dependent economic structure encoded in the payment probabilities (Table~\ref{tab:paymentProbs}). The boundary classes exhibit one-sided interaction patterns (the poorest never pay and the richest never receive), while the middle class both pays and receives. This structural asymmetry provides a plausible explanation for the larger and oppositely signed deviations in classes $j=1$ and $j=3$, and the smaller deviation in class $j=2$.

\begin{figure}[h!]
    \centering
    \includegraphics[width=0.8\textwidth]{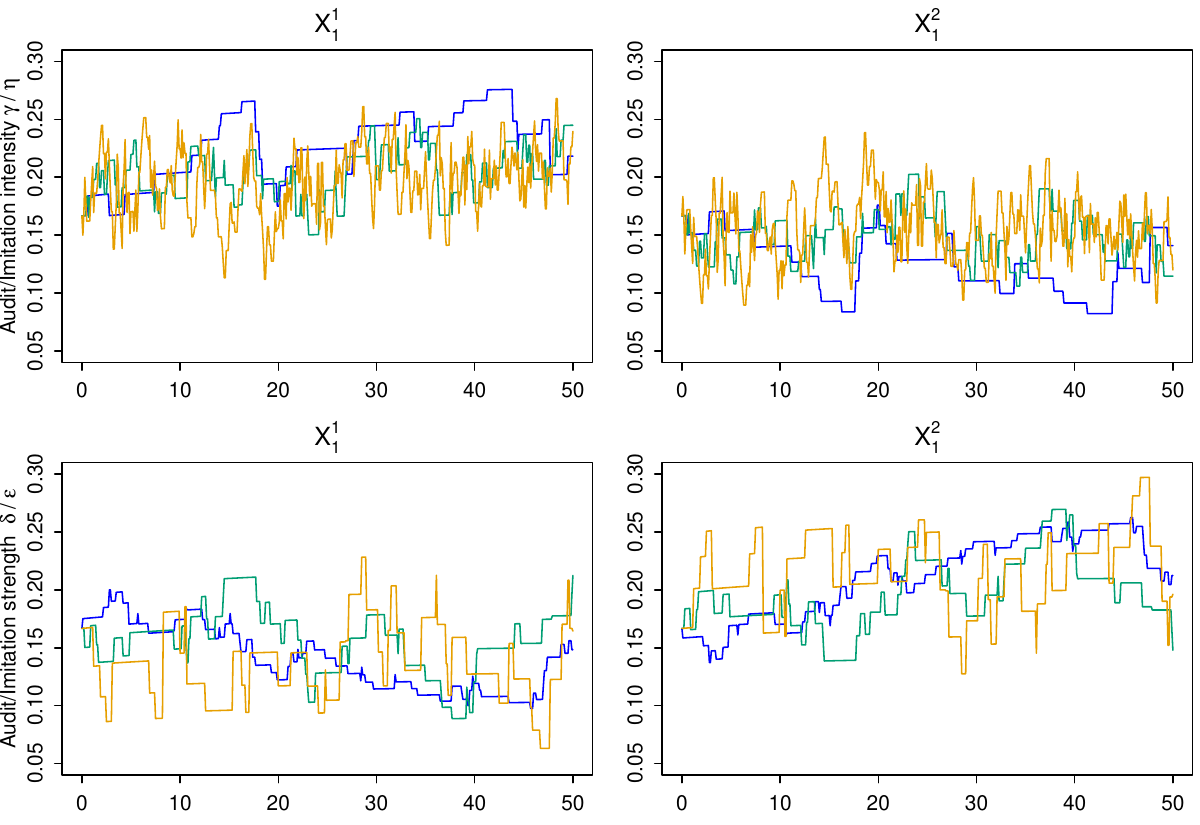}
    \caption{\textbf{Effect of audit and imitation intensities $(\gamma,\eta)$} (top panels). Population fraction paths of income class $1$ for different parameter combinations (blue: $(\gamma,\eta)=(1,0.5)$, green: $(\gamma,\eta)=(3,2)$, orange: $(\gamma,\eta)=(15,10)$). The strength parameters are set to $\delta=\epsilon=0.1$.
    \textbf{Effect of audit and imitation strengths $(\delta,\epsilon)$} (bottom panels). Population fraction paths of income class $1$ for different parameter combinations (blue: $(\delta,\epsilon)=(0.05,0.05)$, green: $(\delta,\epsilon)=(0.1,0.1)$, orange: $(\delta,\epsilon)=(0.2,0.2)$). The intensity parameters are set to $\gamma=1$ and $\eta=2$.}
    \label{fig:intensity_effect}
\end{figure}

\paragraph{Audit and imitation parameters.}

The dynamics of the combined PDMP depend critically on the model parameters $\gamma$, $\eta$, $\delta$ and $\epsilon$ (see Table \ref{tab_params}). In the following, we investigate their impact.

To assess the influence of the event-intensity parameters $\gamma$ (audit) and $\eta$ (imitation), we simulate trajectories for several parameter choices while maintaining fixed values for the other parameters. As shown in Figure~\ref{fig:intensity_effect} (top panels), larger values of $\gamma$ and $\eta$ increase 
the frequency of jump events, consistent with their roles as audit and imitation rate parameters in the 
state-dependent intensity function $\Lambda(x)$ (see \eqref{eq:combined-rate}). The magnitude of individual jumps remains unchanged 
across parameter variations, as jump sizes are determined exclusively by the strength parameters $\delta$ and $\epsilon$, which are held fixed here. 

Now, we vary the event strength parameters $\delta$ (audit) and $\epsilon$ (imitation) while keeping the other parameters fixed.
Figure~\ref{fig:intensity_effect} (bottom panels) demonstrates that increasing $\delta$ and $\epsilon$ amplifies 
the magnitude of population reallocations: higher audit effectiveness $\delta$ transfers larger 
population fractions toward compliance during audit events, while higher imitation strength $\epsilon$ 
yields more substantial shifts toward non-compliance during imitation events. Consequently, the larger these strength parameters, the more pronounced the reallocation movements in the trajectories become.

In summary, the two parameter pairs affect the combined PDMP model in different ways: $(\gamma,\eta)$ determine the frequency of audit and imitation events through the state-dependent jump rate 
$\Lambda(x)$ (see \eqref{eq:combined-rate}), while $(\delta,\epsilon)$ govern the magnitude of the corresponding population transfers 
during each event. Increasing either parameter pair leads to greater variability and 
stronger transitions between sectors, although through different mechanisms.

\section{Conclusion and perspective}
\label{conclusion}

This work extends the tax evasion model of Bertotti and Modanese \cite{bertotti2018mathematical} by embedding it in a piecewise deterministic Markov process (PDMP) framework. The original kinetic system, which describes income‐class transitions through deterministic interactions, was augmented by two stochastic mechanisms that capture essential behavioural responses observed in real economies: audits, which promote compliance, and imitation, which fosters evasion. Each mechanism was formalised as a PDMP with explicitly defined flow, jump rate function, and transition kernel, allowing a rigorous mathematical analysis of their structural properties.

For both audit and imitation models, we established the preservation of fundamental conserved quantities, including total population mass and aggregate income. Simulations confirmed that the two mechanisms induce contrasting long-term behaviours: audit events drive the population toward full compliance, whereas imitation events shift it toward full non-compliance. These contrasting dynamics motivated the construction of a combined model in which both mechanisms act concurrently. The resulting combined PDMP captures the interaction of opposing social and institutional forces and overcomes a key limitation of the deterministic framework, where evasion behaviour is fixed.

The combined PDMP model retains the invariants of the deterministic system but no longer converges to a deterministic equilibrium state. Instead, 
numerical experiments indicate that, under suitable parameter choices, the combined PDMP approaches a stationary distribution supported in a neighbourhood of the deterministic equilibrium. The model therefore offers a richer and more realistic description of tax-evasion dynamics, reflecting how populations continually adjust compliance behaviours in response to oversight and peer effects.

Overall, the PDMP formulation provides a flexible mathematical framework for integrating deterministic income dynamics with stochastic behavioural mechanisms. It opens a variety of possibilities for further extensions, such as stochastic transition kernels, more refined behavioural sectors, alternative audit policies, or learning mechanisms driven by empirical data. The framework may also serve as a basis for studying equilibrium distributions, stability conditions, and control strategies in socio-economic systems where behaviour evolves through a combination of structural incentives and random events.


\appendix
\renewcommand{\appendixpagename}{Appendix}
\appendixpage


\section{Simulation of the PDMP model}
\label{PDMP_simulation}

In this section, we describe how to simulate paths of the proposed PDMP model. We start by presenting a modified Euler method, which preserves the total population (cf. Property \ref{prop1}), for the simulation of the $n\times m$-dimensional ODE defined by \eqref{population_eq}. Then, we embed this method into a thinning procedure for the simulation of the PDMP.

\subsection{ODE simulation: Conserving Euler method}
\label{Im_Eul}

Direct application of the standard Euler method \cite{euler_citation} to the $n\times m$-dimensional ODE defined by~\eqref{population_eq} produces severe numerical instability. In particular, the Euler method does not preserve the population conservation property (see Property \ref{prop1} in Section \ref{Back_sec}), yielding approximate solutions that may drift away from the set
\begin{equation*}
    \Big\{x \in \mathbb{R}^{n \times m} \colon x_j^\alpha \geq 0, \ \sum_{j=1}^n\sum_{\alpha=1}^m x_j^\alpha = 1\Big\}.
\end{equation*}
Figures~\ref{euler}--\ref{euler2} illustrate this non-conserving behaviour: At a certain point in time (here~\mbox{$t \approx 35$}) the total population
generated by the Euler scheme starts to deviate from~$1$ and the resulting trajectory becomes unusable for further analysis.

To obtain a stable numerical method that preserves population conservation, we modify the standard Euler method by employing a  
simple normalisation step after each Euler update. Let~$\tilde{x}(t)$ denote the standard Euler approximation of $x(t)$ at time $t$. The modified scheme replaces~$\tilde{x}(t)$~by 
\begin{equation}\label{eq:mod-euler}
\hat{x}_j^\alpha(t)
    := \frac{\tilde{x}_j^\alpha(t)}
            {\displaystyle\sum_{i=1}^n\sum_{\beta=1}^m \tilde{x}_i^\beta(t)} , \qquad j=1,\ldots,n, \ \alpha=1,\ldots,m,
\end{equation}
ensuring that
\[
\sum_{j=1}^n\sum_{\alpha=1}^m \hat{x}_j^\alpha(t) = 1.
\]
This additional normalisation step enforces preservation of the total population at each time step, yielding reliable approximations $\hat{x}(t)$ of the true solution $x(t)$ of the ODE (see \mbox{Figures~\ref{euler}--\ref{euler2}}).  
We refer to the modified Euler method using \eqref{eq:mod-euler} as the \textit{conserving Euler method}.

\begin{figure}
\centering
\includegraphics[width=1.0\textwidth]{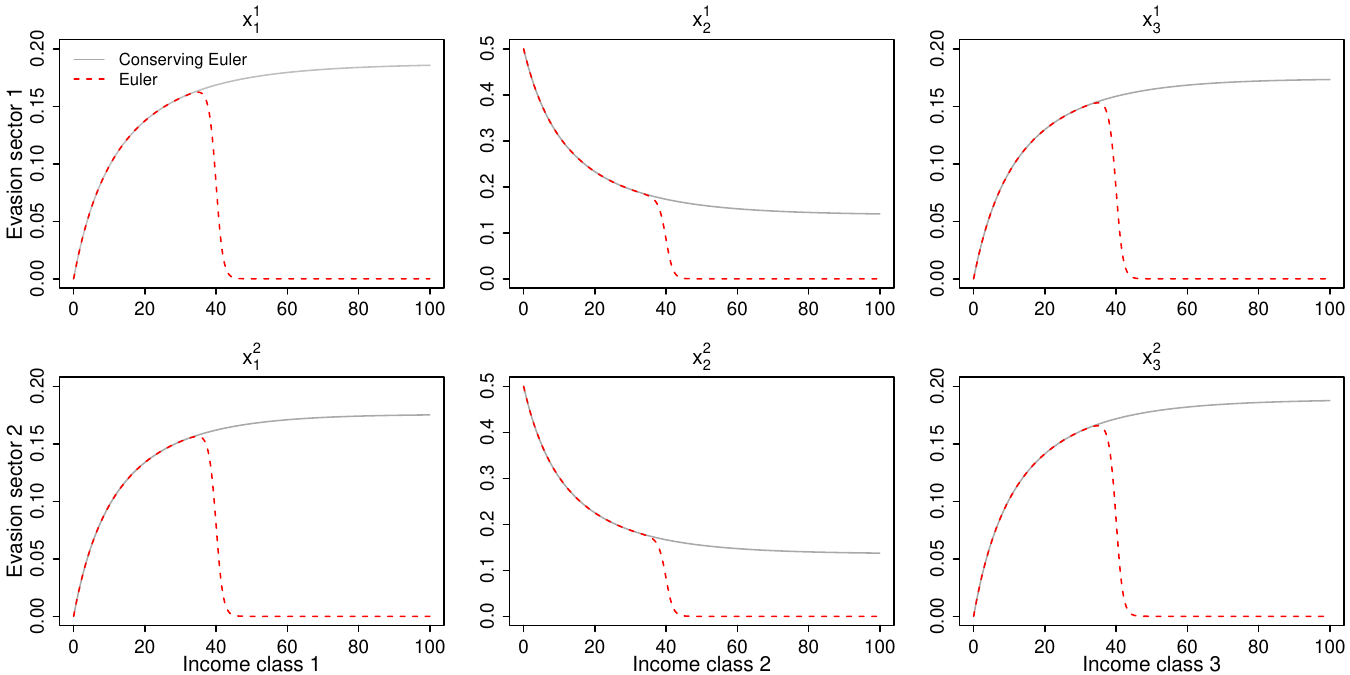}
\caption{Comparison of the standard Euler method (dashed red) and the proposed  conserving Euler method
(solid grey).}
\label{euler}
\end{figure}

\begin{figure}[h!]
\centering
\includegraphics[width=0.6\textwidth]{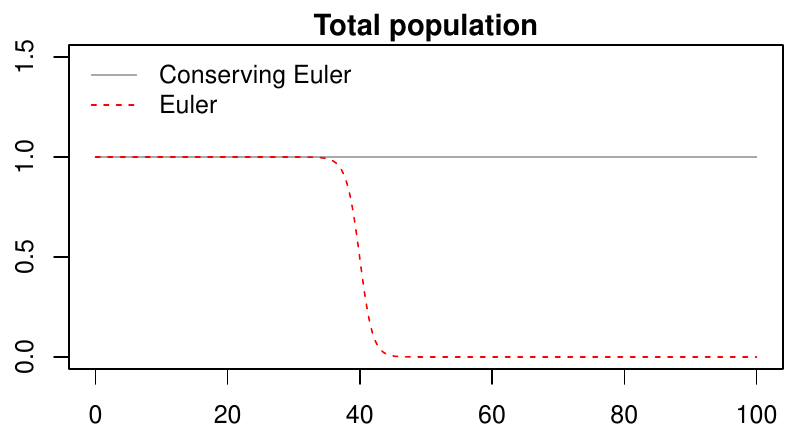}
\caption{Total population corresponding to Figure~\ref{euler} for the standard Euler
method (dashed red) and the proposed conserving Euler method (solid grey).}
\label{euler2}
\end{figure}


\subsection{Thinning method}
\label{Thin_meth}

Let $(X(t))_{t\ge 0}$ be a PDMP with total jump rate 
$\Lambda(x) = \lambda_{\mathrm{audit}}(x) + \lambda_{\mathrm{imitation}}(x)$. Assume that $\lambda_{\mathrm{audit}}(x) \le \gamma$ and 
$\lambda_{\mathrm{imitation}}(x) \le \eta$, so that 
\begin{equation*}
\Lambda(x) \le \bar\Lambda := \gamma + \eta.    
\end{equation*}

\noindent Under this bound, jump times are simulated via thinning of a homogeneous Poisson process with rate $\bar\Lambda$; see \cite{lewis1979simulation}. Let $(\omega_i)_{i\ge1}$ be i.i.d.\ exponential random variables with parameter $\bar\Lambda$, and define the candidate times by
\begin{equation*}
T_i^\ast = T_{i-1} + \omega_i.
\end{equation*}

\noindent At each time $T_i^\ast$, the state is first evolved deterministically up to $T_i^\ast-$, yielding $X({T_i^\ast-})$. The candidate time is then accepted with probability 
$a_i = \Lambda(X({T_i^\ast-}))/\bar\Lambda$. Equivalently, for \mbox{$U_i \sim \mathcal{U}(0,1)$},
\begin{equation*}
T_i =
\begin{cases}
T_i^\ast, & \text{if } U_i \le \dfrac{\Lambda(X({T_i^\ast-}))}{\bar\Lambda},\\
\text{reject}, & \text{otherwise}.
\end{cases}
\end{equation*}

\noindent The accepted times $(T_i)$ form the sequence of jump times of the PDMP and thus generate a point process with stochastic intensity $\Lambda(X(t))$.

\begin{remark}
Note that the above construction is stated for the combined model with total rate $\Lambda(x) = \lambda_{\mathrm{audit}}(x) + \lambda_{\mathrm{imitation}}(x)$. 
In the case where only one mechanism is present, the procedure simplifies accordingly. 
If only audit jumps are considered, then $\Lambda(x) = \lambda_{\mathrm{audit}}(x)$ and the thinning is performed with upper bound $\gamma$. Similarly, for the imitation-only model, one has $\Lambda(x) = \lambda_{\mathrm{imitation}}(x)$ and upper bound $\eta$. In both cases, the acceptance probability reduces to the corresponding single rate divided by its bound.
\end{remark}

\begin{remark}
In the present framework, the rates $\lambda_{\mathrm{audit}}$ and $\lambda_{\mathrm{imitation}}$ are defined so as to be uniformly bounded by the parameters $\gamma$ and $\eta$, respectively. 
In more general PDMP settings, suitable dominating rates may be difficult to obtain and can complicate simulation; see, e.g.,~\mbox{\cite{BERTAZZI202291, lemaire2018exact}}.
\end{remark}

\subsection{PDMP simulation: Thinning and conserving Euler}
\label{Final_PDMP_sim}

In this section, we combine the conserving Euler method introduced in Section~\ref{Im_Eul} with the thinning procedure of Section~\ref{Thin_meth} to simulate sample paths of the PDMP up to a fixed time horizon $T_{\mathrm{final}}$. The proposed method is summarised in Algorithm \ref{alg:PDMP_sim}.

\begin{algorithm}
\caption{Simulation of the PDMP}
\label{alg:PDMP_sim}
\begin{algorithmic}[1]
\State \textbf{Initialise:} $T_0 = 0$, $X({T_0}) = x_0$, choose $T_{\mathrm{final}} > 0$, set $i = 1$
\While{$T_{i-1} < T_{\mathrm{final}}$}
    \State Generate $\omega_i \sim \mathrm{Exp}(\bar\Lambda)$ and set
    \begin{equation*}
        T_i^\ast = T_{i-1} + \omega_i
    \end{equation*}

    \State Evolve the deterministic dynamics on $[T_{k-1}, T_i^\ast)$ using the conserving Euler method, yielding $X({T_i^\ast-})$

    \State Generate $U_i \sim \mathcal{U}(0,1)$

    \If{$U_i \le \dfrac{\Lambda(X({T_i^\ast-}))}{\bar\Lambda}$}
        \State $T_i = T_i^\ast$
        \State Determine the jump type according to
        \[
        \mathbb{P}(\text{audit}) = \frac{\lambda_{\mathrm{audit}}(X({T_i-}))}{\Lambda(X({T_i-}))}, 
        \qquad
        \mathbb{P}(\text{imitation}) = \frac{\lambda_{\mathrm{imitation}}(X({T_i-}))}{\Lambda(X({T_i-}))}
        \]
        \State Sample the post-jump state $X({T_i})$ accordingly
        \State $i \gets i+1$
    \Else
        \State $T_{i-1} \gets T_i^\ast$, \quad $X({T_{k-1}}) \gets X({T_i^\ast-})$
    \EndIf
\EndWhile
\end{algorithmic}
\end{algorithm}


\newpage
\printbibliography

\end{document}